\documentclass{lmcs}
\pdfoutput=1

\usepackage{silence}
\usepackage{lastpage}
\lmcsdoi{15}{3}{21}
\lmcsheading{}{\pageref{LastPage}}{}{}%
{Feb.~23,~2018}{Aug.~27,~2019}{}

\usepackage{color}
\usepackage{calrsfs}
\usepackage{amssymb}

\title{Query Learning of Derived \texorpdfstring{$\omega$}{omega}-Tree Languages in Polynomial Time}

\author[D. Angluin]{Dana Angluin\rsuper{a}}
\author[T. Antonopoulos]{Timos Antonopoulos\rsuper{a}}
\address{\lsuper{a}Yale University, New Haven, CT, USA}
\email{\{dana.angluin,timos.antonopoulos\}@yale.edu}
\author[D. Fisman]{Dana Fisman\rsuper{b}}
\address{\lsuper{b}Ben-Gurion University, Be'er Sheva, Israel}
\email{dana@cs.bgu.ac.il}

\usepackage{mathbbol}
\usepackage{float}
\usepackage{algorithm}
\usepackage[noend]{algorithmic}
\usepackage[font={small,it}]{caption}
\usepackage{wrapfig}
\usepackage{tikz}
\usetikzlibrary{arrows,positioning,shapes,decorations,automata,backgrounds,petri,fit,calc}

\newcommand{\concept}[1]{\emph{#1}}

\newcommand{\commentout}[1]{{}}

\newcommand{\class}[1]{\ensuremath{\mathbb{#1}}}

\newcommand{\dollar}{{\$}}
\newcommand{\ldollar}{\ensuremath{L_{\$}}}
\newcommand{\lstar}{\ensuremath{{L^*}}}

\DeclareMathOperator{\dfw}{\class{DFW}}
\DeclareMathOperator{\nfw}{\class{NFW}}
\DeclareMathOperator{\dbw}{\class{DBW}}
\DeclareMathOperator{\dcw}{\class{DCW}}
\DeclareMathOperator{\dpw}{\class{DPW}}
\DeclareMathOperator{\dwbw}{\class{DwBW}}
\DeclareMathOperator{\dwcw}{\class{DwCW}}
\DeclareMathOperator{\dwpw}{\class{DwPW}}

\DeclareMathOperator{\dbwdcw}{\dbw\cap\dcw}
\DeclareMathOperator{\nbw}{\class{NBW}}

\DeclareMathOperator{\npw}{\class{NPW}}
\DeclareMathOperator{\dbt}{\class{DBT}}
\DeclareMathOperator{\nbt}{\class{NBT}}

\DeclareMathOperator{\npt}{\class{NPT}}

\DeclareMathOperator{\C}{\class{C}}

\newcommand{\mtd}{M^{T,d}}

\newcommand{\spfunc}[1]{\ensuremath{\textsl{#1}}}

\DeclareMathOperator{\infset}{\spfunc{Inf}}

\DeclareMathOperator{\Trees}{\spfunc{Trees}}
\DeclareMathOperator{\paths}{\spfunc{paths}}
\DeclareMathOperator{\transitions}{\spfunc{transitions}}
\DeclareMathOperator{\tree}{\spfunc{tree}}
\DeclareMathOperator{\acceptor}{\spfunc{acceptor}}

\newcommand{\procname}[1]{\ensuremath{\textsl{#1}}}

\DeclareMathOperator{\Acc}{\procname{Accepted?}}
\DeclareMathOperator{\Nextsymbol}{\procname{Nextsymbol}}
\DeclareMathOperator{\Nextword}{\procname{Nextword}}
\DeclareMathOperator{\Findperiod}{\procname{Findperiod}}
\DeclareMathOperator{\Findprefix}{\procname{Findprefix}}
\DeclareMathOperator{\Findctrex}{\procname{Findctrex}}

\newcommand{\query}[1]{\textsc{#1}}

\DeclareMathOperator{\MQ}{\query{mq}}
\DeclareMathOperator{\EQ}{\query{eq}}
\DeclareMathOperator{\RSQ}{\query{rsq}}
\DeclareMathOperator{\USQ}{\query{usq}}

\newcommand{\alg}[1]{{\mbox{${\mathbf{#1}}$}}}

\DeclareMathOperator{\A}{\alg{A}}
\DeclareMathOperator{\R}{\alg{R}}
\newcommand{\ATrees}{\A_{{\Trees}}}

\newcommand{\lang}[1]{{\mathbb{\Lbrack}}{#1}{\mathbb{\Rbrack}}}

\begin{document}

\maketitle

\begin{abstract}
We present the first polynomial time algorithm to learn
nontrivial classes of languages of infinite trees.
Specifically, our algorithm uses membership and equivalence queries
to learn classes of $\omega$-tree languages
derived from weak regular $\omega$-word languages
in polynomial time.
The method is a general polynomial time reduction
of learning a class of derived $\omega$-tree languages to
learning the underlying class of $\omega$-word languages,
for any class of $\omega$-word languages recognized by a deterministic B\"{u}chi acceptor. %in $\dbw$.
Our reduction, combined with the polynomial time
learning algorithm of Maler and Pnueli~\cite{Maler1995}
for the class of weak regular $\omega$-word languages
yields the main result.
We also show that subset queries that return counterexamples
can be implemented in polynomial time using subset queries
that return no counterexamples for deterministic or
non-deterministic finite word
acceptors, and deterministic or non-deterministic B\"{u}chi
$\omega$-word acceptors.

A previous claim of an algorithm to learn regular $\omega$-trees
due to Jayasrirani, Begam and Thomas~\cite{Jayasrirani:2008}
is unfortunately incorrect, as shown in~\cite{Angluin:2016}.
\end{abstract}

\section{Introduction}%
\label{section-introduction}

Query learning is a framework in which a learning algorithm attempts to identify a target concept using specified types of queries to an oracle (or teacher) about the target concept~\cite{Angluin:1988}. For example, if the target concept is a regular language $L$, a membership query asks whether a string $x$ is a member of $L$, and is answered either ``yes'' or ``no''. An equivalence query asks whether a hypothesis language $L'$ (represented, for example, by a deterministic finite acceptor) is equal to $L$. In the case of an equivalence query, the answer may be ``yes'', in which case the learning algorithm has succeeded in exact identification of the target concept, or it may be ``no'', accompanied by a counterexample, that is, a string $x$ in $L$ but not in $L'$ (or vice versa). The counterexample is a witness that $L'$ is not equal to $L$.

When $L'$ is not equal to $L$, there is generally a choice (often an infinity) of possible counterexamples, and we require that the learning algorithm works well regardless of which counterexample is chosen by the teacher. To account for this in terms of quantifying the running time of the learning algorithm, we include a parameter that is the maximum length of any counterexample returned by the teacher at any point in the learning process. In this setting, the \lstar\ algorithm of Angluin~\cite{Angluin87} learns any regular language $L$ using membership and equivalence queries in time polynomial in the size of the smallest deterministic finite acceptor for $L$ and the length of the longest counterexample chosen by the teacher. As shown in~\cite{Angluin:1990}, there can be no such polynomial time algorithm using just membership queries or just equivalence queries.

The assumption that equivalence queries are available may seem unrealistic.  How is a person or a program to judge the equivalence of the target concept to some large, complex, technical specification of a hypothesis? If the hypothesis and the target concept are both deterministic finite acceptors, there is a polynomial time algorithm to test equivalence and return a counterexample in case the answer is negative. Alternatively, if there is a polynomial time algorithm for exact learnability of a class $\C$ of concepts using membership and equivalence queries, then it may be transformed into a polynomial time algorithm that learns approximations of concepts from $\C$ using membership queries and randomly drawn labeled examples~\cite{Angluin87,Angluin:1988}.  In this transformation, there is an unknown probability distribution on examples, and an approximation bound $\epsilon > 0$ and a confidence bound $\delta > 0$ are given, and the algorithm draws a corpus of labeled examples of cardinality polynomial in the size of the target concept, $1/\epsilon$ and $\log(1/\delta)$.
To answer an equivalence query, the hypothesis is checked against the labeled examples in the corpus.
If the hypothesis agrees with the labels of all the examples in the corpus, the equivalence query is answered
``yes'', and otherwise, any exception supplies a counterexample to return as the answer of the equivalence
query.  The final hypothesis output by the transformed algorithm will, with probability at least $1-\delta$, have a probability of at most $\epsilon$ of disagreeing with the target on examples drawn from the unknown probability distribution.

Since the publication of $L^*$, there have been a number of substantial improvements and extensions of the algorithm, as well as novel and unanticipated applications in the analysis, verification and synthesis of programs, protocols and hardware, following the work of Peled et al.\ that identified the applicability of
$L^*$ in the area of formal methods~\cite{PeledVY02}. In a recent CACM review article, Vaandrager~\cite{Vaandrager:2017} explains Model Learning, which takes a black box approach to learning a finite state model of a given hardware or software system using membership queries (implemented by giving the system a sequence of inputs and observing the sequence of outputs) and equivalence queries (implemented using a set of test sequences in which the outputs of the hypothesis are compared with the outputs of the given system.) The learned models may then be analyzed to find discrepancies between a specification and its implementation, or between different implementations. He cites applications in telecommunications~\cite{HagererHNS02,ShahbazG14}, the automotive industry~\cite{FengLMNSW13}, online conference systems~\cite{WindmullerNSHB13}, as well as analyzing botnet protocols~\cite{ChocSS10}, smart card readers~\cite{ChaluparPPR14}, bank card protocols~\cite{AartsRP13}, network protocols~\cite{RuiterP15} and legacy software~\cite{MargariaNRS04,SchutsHV16}.

Another application of finite state machine learning algorithms is in the assume-guarantee approach to verifying systems by dividing them into modules that can be verified individually. Cobleigh, Giannakopoulou and P\u{a}sare\u{a}nu~\cite{Cobleigh:2003} first proposed using a learning algorithm to learn a correct and sufficient contextual assumption for the component being verified, and there has since been a great deal of research progress in this area~\cite{NamA06}.

If we consider {\it reactive systems}, that is, systems that maintain an ongoing interaction with their environment (e.g., operating systems, communication protocols, or robotic swarms), the restriction to models specified by finite automata processing finite sequences of inputs is too limiting. Instead, one trajectory of the behavior of a reactive system may be modeled using an infinite word ($\omega$-word), each symbol of which specifies the current state of the system and the environment at a given time. The system itself may be modeled by an $\omega$-automaton, that is, a finite state automaton that processes $\omega$-words. The desired behavior of such a system may be specified by a linear temporal logic formula, that defines the set of $\omega$-words that constitute ``good'' behaviors of the system.

Researchers have thus sought query learning algorithms for $\omega$-automata that could be used in the settings of model learning or assume-guarantee verification for reactive systems. However, learning $\omega$-automata seems to be a much more challenging problem than learning automata on finite words, in part because the Myhill-Nerode characterization for regular languages (stating that there is a unique minimum deterministic acceptor that can be constructed using the right congruence classes of the language) does not hold in general for regular $\omega$-languages. The Myhill-Nerode characterization is the basis of the \lstar\ algorithm and its successors.

There is no known polynomial time algorithm using membership and equivalence queries to learn even the whole class $\dbw$ of languages recognized by deterministic B\"{u}chi acceptors, which is a strict subclass of the class of all regular $\omega$-languages. Maler and Pnueli~\cite{Maler1995} have given a polynomial time algorithm using membership and equivalence queries to learn the {weak regular $\omega$-languages}. This class, denoted $\dwpw$, is the set of languages accepted by deterministic weak parity automata, and is a non-trivial subclass of  $\dbw$. The class $\dwpw$ does have a Myhill-Nerode property, but this alone does not suffice for extending \lstar\ to learn this class, since the observed data might suggest conflicting ways to mark accepting states in an automaton agreeing with the observed data. Maler and Pnueli's algorithm manages to overcome this problem by finding a set of membership queries to ask to resolve the conflict.

In the context of assume-guarantee verification, Farzan et al.~\cite{FarzanCCTW08} proposed a direct application of \lstar\ to learn the full class of regular $\omega$-languages.  Their approach is based on the result of Calbrix, Nivat and Podelski~\cite{CalbrixNP93} showing that a regular $\omega$-language $L$ can be characterized by the regular language \ldollar\ of finite strings $u \dollar v$ representing the set of ultimately periodic words ${u(v)}^{\omega}$ in $L$.  This establishes that a regular $\omega$-language $L$ is learnable using membership and equivalence queries in time polynomial in the size of the minimal deterministic finite acceptor for \ldollar. However, the size of this representation may be exponentially larger than its $\omega$-automaton representation. More recently, Angluin and Fisman~\cite{AngluinF16} have given a learning algorithm using membership and equivalence queries for general regular $\omega$-languages represented by families of deterministic finite acceptors, which improves on the \ldollar\
representation, however the running time is not bounded by a polynomial
in the representation.  Clearly, much more research is needed in the area of query learning of $\omega$-automata.

Despite the difficulties in learning $\omega$-automata, which are used in the analysis of linear temporal logic, in this paper we consider the theoretical question of learning $\omega$-tree automata, which are used in the analysis of branching temporal logic~\cite{EmersonS88,KupfermanVW00}.
As a potential motivation for studying learning of $\omega$-tree languages, we
consider a setting in which two players play an infinite game in which
the opponent chooses one of two actions ($1$ and $2$) and the player responds with
a symbol chosen from a finite alphabet $\Sigma$.
We can represent the strategy of the player as a binary $\omega$-tree in which each node is the player's state, the two edges leaving the node are the possible choices of the opponent (action $1$ or $2$), each edge is labeled with the response action (from $\Sigma$) of the player, and each leads to a (potentially new) state for the player.
In this interpretation, a set of $\omega$-trees represents a property of strategies, and the task of the learner is to learn an initially unknown property of strategies by using membership queries (``Does this strategy have the unknown property?'') and equivalence queries (``Is this property the same as the unknown property of strategies?'') answered either ``yes'' or with a counterexample, that is, a strategy that distinguishes the two properties.

Because of the difficulty of the problem, we restrict our attention to $\omega$-tree languages such that all of their paths satisfy a certain temporal logic formula, or equivalently, a property of $\omega$-words. Given an $\omega$-word language $L$, we use $\Trees_d(L)$ to denote the set of all $d$-ary $\omega$-trees $t$ all of whose paths are in $L$. The $\omega$-tree language $\Trees_d(L)$ is often referred to as the \emph{derived} language of $L$. In this context, it is natural to ask whether learning a derived $\omega$-tree language $\Trees_d(L)$ can be reduced to learning the $\omega$-word language $L$.

We answer this question affirmatively for the case that $L$ can be recognized by a deterministic B\"{u}chi word automaton and learned using membership and equivalence queries. Applying this reduction to the result of Maler and Pnueli on polynomially learning languages in $\dwpw$ we obtain a polynomial learning algorithm for derived languages in $\Trees_d(\dwpw)$ using membership and equivalence queries. Moreover, any progress on polynomially learning an extended subclass $\C$ of $\dbw$ using membership and equivalence queries can be automatically lifted to learning $\Trees_d(\C)$.

%%%%%%%%%%%%%%%%%%%%%%%%%%%%%%%%%%%%%%%%%%%%%%%%%

The framework of the reduction is depicted in Fig.~\ref{fig:reduction-framework}. An  algorithm $\ATrees$ for learning $\Trees_d(L)$ uses a learning algorithm $\A$ for $L$ to complete its task. The membership and equivalence queries ($\MQ$ and $\EQ$, henceforth) of algorithm  $\ATrees$ are answered by respective oracles $\MQ$ and $\EQ$ for $\Trees_d(L)$. Since $\A$ asks membership and equivalence queries about $L$ rather than $\Trees_d(L)$, the learner $\ATrees$ needs to find a way to answer these queries. If $\A$ asks a membership query about an $\omega$-word, $\ATrees$ can ask a membership query about an $\omega$-tree all of whose paths are identical to the given $\omega$-word. Since the tree is accepted by $\Trees_d(L)$ iff the given word is accepted by $L$ it can pass the answer as is to $\A$. If $\A$ asks an equivalence query, using an acceptor $M$ for an $\omega$-language, then $\ATrees$ can ask an equivalence query using an acceptor $M^T$ that accepts an $\omega$-tree if all its paths are accepted by $M$. If this query is answered positively then $\ATrees$ can output the tree acceptor $M^T$ and halt. The challenge starts when this query is answered negatively.

When the $\EQ$ is answered negatively, a counterexample tree $t$ is given. There are two cases to consider. Either this tree is in $\Trees_d(L)$ but is rejected by the hypothesis acceptor $M^T$, in which case $t$ is referred to as a \emph{positive counterexample}; or this tree is not in $\Trees_d(L)$ but is accepted by the hypothesis acceptor $M^T$, in which case $t$ is referred to as a \emph{negative counterexample}.
If $t$ is a positive counterexample, since $M^T$ rejects $t$ there must be a path in $t$ which is rejected by $M$. It is not too dificult to extract that path. The real challenge is dealing with a negative counterexample. This part is grayed out in the figure. In this case the tree $t$ is accepted by $M^T$ yet it is not in $\Trees_d(L)$. Thus, all the paths of the tree are accepted by $M$, yet at least one path is not accepted by $L$. Since $L$ is not given, it is not clear how we can extract such a path. Since we know that not all paths of $t$ are contained in $L$, a use of an unrestricted subset query could help us. Unrestricted subset queries ($\USQ$) are queries on the inclusion of a current hypothesis in the unknown language that are answered by ``yes'' or ``no''  with an accompanying counterexample in the case the answer is negative.

\begin{figure}
%\begin{wrapfigure}{r}{0.5\textwidth}
    \centering
	\scalebox{0.30}{
		\includegraphics{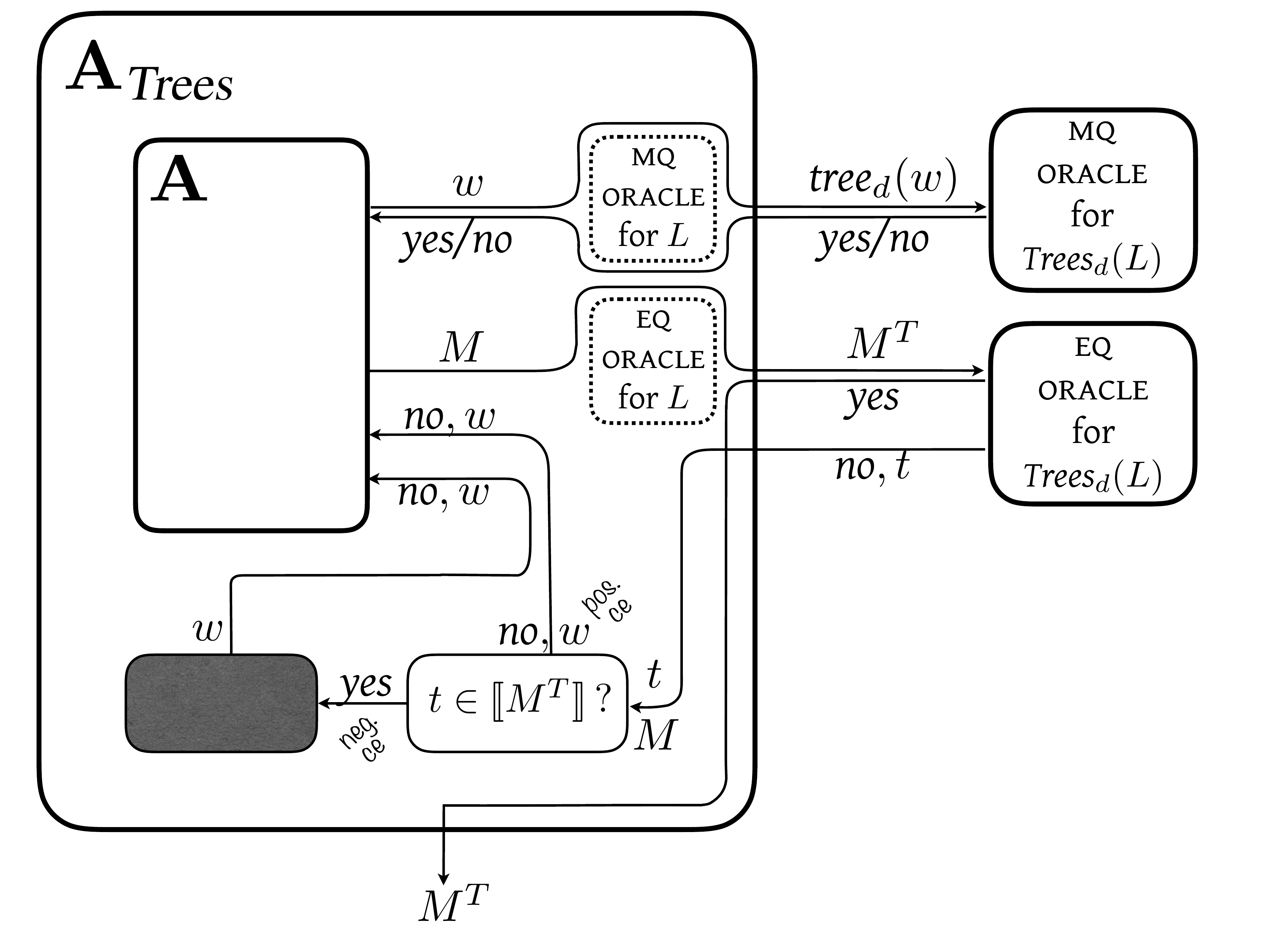}
	}
	\caption{The reduction framework}\label{fig:reduction-framework}
%\end{wrapfigure}
\end{figure}

Since we don't have access to $\USQ$s we investigate whether we can obtain such queries given the queries we have. We show that unrestricted subset queries can be simulated by restricted subset queries.  Restricted subset queries ($\RSQ$) on $\omega$-words are subset queries that are \emph{not} accompanied by counterexamples. This essentially means that there is a way to construct a desired counterexample without it being given.  To discharge the use of restricted subset queries (as the learner is not provided such queries either) we investigate the relation between subsets of $\omega$-words and $\omega$-trees. Finally, we show that the desired subset queries on $\omega$-words can be given to the $\omega$-tree learning algorithm by means of $\omega$-tree membership queries. From these we can construct a series of procedures to implement the grayed area.

The subsequent sections contain definitions of $\omega$-words, $\omega$-trees and automata processing them, derived $\omega$-tree languages, the problem of learning classes of $\omega$-word and $\omega$-tree languages, preliminary results, the algorithm for the main reduction, and some discussion. We also include an Appendix with a few examples illustrating some of the procedures involved in our framework.
% Some proof details are deferred to the Appendix.

\section{Definitions}%
\label{section-definitions}

\subsection{Words and trees}
(For more details
see Gr\"{a}del, Thomas and Wilke~\cite{Gradel2002},
Perrin and Pin~\cite{PP2004}, and
L\"{o}ding~\cite{Loding11}.)
Let $\Sigma$ be a fixed finite alphabet of symbols.
The set of all finite words over $\Sigma$ is denoted $\Sigma^*$.
The empty word is denoted $\varepsilon$, and the length of a finite
word $x$ is denoted $|x|$.
$\Sigma^+$ is the set of all nonempty finite words over $\Sigma$,
and for a nonnegative integer $k$, $\Sigma^k$ is the set of
all finite words over $\Sigma$ of length equal to $k$.
A finite word language is a subset of $\Sigma^*$.

An $\omega$-word over $\Sigma$ is an infinite sequence
$w = \sigma_1 \sigma_2 \sigma_3 \cdots$
where each $\sigma_i \in \Sigma$.
The set of all $\omega$-words over $\Sigma$ is denoted $\Sigma^{\omega}$.
An $\omega$-word language is a subset of $\Sigma^{\omega}$.
The $\omega$-regular expressions are analogous to
finite regular expressions, with the added operation $S^{\omega}$,
where $S$ is a set of finite words, and the restriction that
concatenation combines a set of finite words as the left argument with a
set of finite words or $\omega$-words as the right argument.
The set $S^{\omega}$ is the set of all $\omega$-words $s_1 s_2 \cdots$ such
that for each $i$, $s_i \in S$ and $s_i \neq \varepsilon$.
For example, ${(a+b)}^* {(a)}^{\omega}$ is the set of all $\omega$-words
over $\{a,b\}$ that contain finitely many occurrences of $b$.

If $S \subseteq \Sigma^*$, $n$ is a nonnegative integer and
$u \in \Sigma^*$, we define the
\concept{length and prefix restricted} version of $S$ by
 $S[n,u] = S \cap \Sigma^n \cap (u \cdot \Sigma^*)$.
This is the set of all words in $S$ of length $n$ that begin with the
prefix $u$.
We also define the \concept{length restricted} version of $S$
by $S[n] = S[n,\varepsilon]$, that is, the set of all
words in $S$ of length $n$.

Let $d$ be a positive integer.
We consider $T_d$, the unlabeled complete $d$-ary $\omega$-tree
whose directions are specified by $D = \{1,\ldots, d\}$.
The \concept{nodes} of $T_d$ are the elements of $D^*$.
The \concept{root} of $T_d$ is the node $\varepsilon$,
and the \concept{children} of node $v$ are $v \cdot i$ for $i \in D$.
An \concept{infinite path} $\pi$ in $T_d$
is a sequence $x_0, x_1, x_2, \ldots$ of nodes
of $T_d$ such that $x_0$ is the root
and for all nonnegative integers $n$,
$x_{n+1}$ is a child of $x_n$.
An infinite path in $T_d$ corresponds to an $\omega$-word over $D$
giving the sequence of directions traversed by the path starting
at the root.

A labeled $d$-ary $\omega$-tree (or just \concept{$\omega$-tree})
is given by a mapping
$t: D^+ \rightarrow \Sigma$
that assigns a symbol in $\Sigma$ to each non-root node of $T_d$.
We may think of $t$ as assigning the symbol $t(v \cdot i)$ to the
tree edge from node $v$ to its child node $v \cdot i$.
The set of all labeled $d$-ary
$\omega$-trees is denoted $T_d^{\Sigma}$.
An $\omega$-tree language is a subset of $T_d^{\Sigma}$.
If $\pi = x_0, x_1, x_2, \ldots$ is an infinite path of $T_d$,
then we define $t(\pi)$ to be the $\omega$-word
$t(x_1), t(x_2), \ldots$ consisting of the sequence of
labels of the non-root nodes of $\pi$ in $t$.
(Recall that $t$ does not label the root node.)

\subsection{Automata on words}%
\label{subsection-Automata-on-words}

A \concept{finite state word automaton}
is given by a tuple $M = (Q,q_0,\delta)$,
where $Q$ is a finite set of states, $q_0 \in Q$ is the initial
state, and $\delta: Q \times \Sigma \rightarrow 2^Q$ is the
(nondeterministic) transition function.
The automaton is \concept{deterministic} if $\delta(q,\sigma)$
contains at most one state for every $(q,\sigma) \in Q \times \Sigma$,
and \concept{complete} if $\delta(q,\sigma)$ contains at least
one state for every $(q,\sigma) \in Q \times \Sigma$.
For a complete deterministic automaton we extend $\delta$ to map $Q \times \Sigma^*$ to $Q$
in the usual way.

Let $x = \sigma_1 \sigma_2 \cdots \sigma_k$ be a finite word,
where each $\sigma_n \in \Sigma$.
A \concept{run} of $M$ on $x$ is a sequence of $k+1$ states
$r_0, r_1, \ldots, r_k$ such that $r_0 = q_0$ is the initial state
and $r_n \in \delta(r_{n-1},\sigma_n)$ for integers $1 \le n \le k$.
Let $w = \sigma_1 \sigma_2 \cdots$ be an $\omega$-word, where
each $\sigma_n \in \Sigma$.
A \concept{run} of $M$ on $w$ is an infinite sequence of
states $r_0, r_1, r_2, \ldots$ such that $r_0 = q_0$ is
the initial state and $r_n \in \delta(r_{n-1},\sigma_n)$ for
all positive integers $n$.

A \concept{nondeterministic finite acceptor} is given by
$M = (Q,q_0,\delta,F)$,
where $(Q,q_0,\delta)$ is a
finite state word automaton,
and the new component $F \subseteq Q$ is the set of
accepting states.
$M$ is a \concept{deterministic finite acceptor} if $\delta$
is deterministic.
Let $M$ be a nondeterministic finite acceptor and
$x \in \Sigma^*$ a finite word of length $n$.
$M$ \concept{accepts} $x$
iff there is a run $r_0, r_1, \ldots, r_n$ of $M$ on $x$
such that $r_n \in F$.
The language \concept{recognized} by $M$ is the set of all
finite words accepted by $M$, denoted by $\lang{M}$.
The class of all finite word languages recognized by
deterministic finite acceptors is denoted by $\dfw$,
and by nondeterministic finite acceptors, $\nfw$.
These representations are equally expressive, that is,
$\nfw = \dfw$.

Turning to finite state word automata processing $\omega$-words,
a variety of different acceptance criteria have
been considered.
Such an acceptor is given by a tuple $M = (Q,q_0,\delta,\alpha)$,
where $(Q,q_0,\delta)$ is a finite state word automaton
and $\alpha$ specifies a mapping from $2^Q$ to $\{0,1\}$
which gives the criterion of acceptance for an
$\omega$-word $w$.

Given an $\omega$-word $w$ and a run $r = r_0, r_1, \ldots$ of
$M$ on $w$, we consider the set
$\infset(r)$ of all states $q \in Q$ such that
$r_n = q$ for infinitely many indices $n$.
The acceptor $M$ \concept{accepts} the $\omega$-word $w$
iff there exists a run $r$ of $M$ on $w$ such
that $\alpha(\infset(r)) = 1$.
That is, $M$ accepts $w$ iff there exists
a run of $M$ on $w$ such that
the set of states visited infinitely often in the
run satisfies the acceptance criterion $\alpha$.
The language \concept{recognized} by $M$ is the set
of all $\omega$-words accepted by $M$, denoted $\lang{M}$.

%A Muller acceptor has the most general acceptance criterion:
%$\alpha$ is specified by giving a set $\cal F$ of subsets
%of $Q$ and defining $\alpha(S) = 1$ iff $S \in {\cal F}$.
For a B\"{u}chi acceptor, the acceptance criterion $\alpha$ is specified by giving
a set $F \subseteq Q$ of accepting states and defining
$\alpha(S) = 1$ iff $S \cap F \neq \emptyset$.
In words, a B\"{u}chi acceptor $M$ accepts $w$ if and only
if there exists a run $r$ of $M$ on $w$ such that at least
one accepting state is visited infinitely often in $r$.
For a co-B\"{u}chi acceptor, the acceptance criterion $\alpha$ is specified by giving
a set $F \subseteq Q$ of rejecting states and defining
$\alpha(S) = 1$ iff $S \cap F = \emptyset$.
For a parity acceptor, $\alpha$ is specified by giving
a function $c$ mapping $Q$ to an interval of integers $[i,j]$,
(called \concept{colors} or \concept{priorities})
and defining $\alpha(S) = 1$
iff the minimum integer in $c(S)$ is even.

%Each of the acceptors has also a \concept{weak} variant. A weak acceptor with acceptance criterion $\alpha$ accepts the $\omega$-word $w$
%iff there exists a run $r$ of $M$ on $w$ such that $\alpha(\occset(r)) = 1$ where $\occset$ is the set of states that \emph{occur} during the run $r$.

A \emph{parity} automaton is said to be \concept{weak} if no two strongly connected states have distinct colors, i.e., if looking at the partition of its states to maximal strongly connected components (MSCCs) all states of an MSCC have the same color. Clearly every weak parity automaton can be colored with only two colors, one even and one odd, in which case the colors are often referred to as \emph{accepting} or \emph{rejecting}. It follows that a weak parity automaton can be regarded as either a B\"uchi or a coB\"uchi automaton.
If in addition no rejecting MSCC is reachable from an accepting MSCC, the acceptor is said to be \emph{weak B\"uchi}. Likewise, a weak parity acceptor where  no accepting MSCC is reachable from a rejecting MSCC,  is said to be \emph{weak coB\"uchi} acceptor.

The classes of languages of $\omega$-words recognized by
these kinds of acceptors will be denoted by three/four-letter
acronyms, with \class{N} or \class{D} (for nondeterministic or deterministic),
\class{B}, \class{C},  \class{P}, \class{wB}, \class{wC} or  \class{wP} (for B\"{u}chi, co-B\"{u}chi,  parity or
their respective weak variants)
and then \class{W} (for $\omega$-words).
Thus $\dwbw$ is the class of $\omega$-word languages recognized by deterministic weak B\"{u}chi word acceptors.
% and $\nmw$ is the class of $\omega$-word languages recognized by nondeterministic Muller word acceptors

\begin{wrapfigure}{4}{0.3\textwidth}
	\centering
	\vspace{-8mm}
	\begin{minipage}[t]{3cm}\centering

		\begin{center}

			\vspace{-1mm}
			\scalebox{0.7}{%
				\begin{tikzpicture}[->,>=stealth',shorten >=1pt,auto,node distance=1.5cm,semithick,initial text=]

				\node[label] (Reactivity)    {$\dpw$};
				\node[label] (Recurrence)    [below left of=Reactivity]{$\dbw$};
				\node[label] (Persistence)    [below right of=Reactivity]{$\dcw$};
				\node[label] (Obligation)    [below of=Reactivity, node distance=2.15cm]{$\dwpw$};
				\node[label] (Guarantee)    [below left of=Obligation]{$\dwbw$};
				\node[label] (Safety)    [below right of=Obligation]{$\dwcw$};

				\path (Safety) edge (Obligation);
				\path (Guarantee) edge (Obligation);

				\path (Recurrence) edge (Reactivity);
				\path (Persistence) edge (Reactivity);

				\path (Obligation) edge (Recurrence);
				\path (Obligation) edge (Persistence);

				\end{tikzpicture}
                }
		\end{center}
	\end{minipage}
	\caption{Expressiveness hierarchy of $\omega$-acceptors~\cite{Wagner75,MP89}.}\label{fig-a-acc-hierearchy}
\end{wrapfigure}
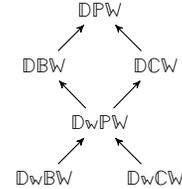

Concerning the expressive power of various types of acceptors,
previous research has established the following results.
The weak variants are strictly less expressive than the non-weak variants.
Deterministic parity automata are more expressive than deterministic B\"uchi and coB\"uchi automata and the same
is true for their weak variants.
These results are summarized in Fig.~\ref{fig-a-acc-hierearchy}. In addition, $\nbw=\dpw=\npw$ and $\dwpw=\dcw \cap \dbw$.
The class of  \emph{regular $\omega$-languages} is the class $\dpw$, and the class of \emph{weak regular $\omega$-languages} is the  class $\dwpw$.

%The class of regular $\omega$-word languages is
%$ \npw = \dpw = \nbw \supseteq \dbw$.
%Thus, nondeterminism does not increase the expressive power of
 %parity $\omega$-word acceptors.
%This is not true of B\"{u}chi or co-B\"{u}chi $\omega$-word acceptors:
%$\dbw$ and $\dcw$ are incomparable, and both are proper subclasses
%of $\nbw$.
%Also of interest is the class $\dbwdcw$, consisting of all
%$\omega$-word languages that are recognized both by
%a deterministic B\"{u}chi acceptor
%and by
%a deterministic co-B\"{u}chi acceptor; this is a proper
%subclass of both $\dbw$ and $\dcw$.

\subsection{Automata on trees}%
\label{subsection-Automata-on-trees}

Acceptors on $d$-ary $\omega$-trees are equipped with analogous
accepting conditions.
Such an acceptor is given by a tuple
$M = (Q,q_0,\delta,\alpha)$, where
$Q$ is a finite set of states, $q_0 \in Q$ is the initial
state, the transition function $\delta$ is a map from $Q$ and $d$-tuples
of symbols to sets of $d$-tuples of states, that
is, $\delta: Q \times \Sigma^d \rightarrow 2^{Q^d}$,
and the acceptance criterion
$\alpha$ specifies a function from $2^Q$ to $\{0,1\}$.
We may think of the acceptor as running top down from the
root of the tree, at each node nondeterministically choosing
a permissible $d$-tuple of states for the $d$ children
of the node depending on the state assigned to the node
and the $d$-tuple of symbols on its outgoing edges. In other words, for each node, with a state $q$ assigned to it, and $d$ outgoing edges with symbols $\sigma_1,\ldots,\sigma_d$, the acceptor will assign states $q_1,\ldots, q_d$ to the children of the nodes, only if $(q_1,\ldots,q_d)\in\delta(q, (\sigma_1,\ldots,\sigma_d))$.

We define a \concept{run} of $M$ on the $\omega$-tree $t$
as a mapping $r$ from the nodes of $T_d$ to $Q$ such that
$r(\varepsilon) = q_0$
and for every node $x$, we have
$(r(x \cdot 1), \ldots, r(x \cdot d)) \in
\delta(r(x), (t(x \cdot 1), \ldots, t(x \cdot d)))$.
That is, the root is assigned the initial state and for
every node,
the ordered $d$-tuple of states assigned to its children
is permitted by the transition function.
The acceptor $M$ \concept{accepts} the $\omega$-tree $t$
iff there exists a run $r$ of $M$ on $t$ such
that for every infinite path $\pi$, we have
$\alpha(\infset(r(\pi))) = 1$.
That is, there must be at least one run in which, for every
infinite path, the set of states that occur infinitely often
on the path satisfies the acceptance criterion $\alpha$.
The $\omega$-tree language \concept{recognized} by $M$ is the
set of all $\omega$-trees accepted by $M$, denoted $\lang{M}$.

The specification of the acceptance criterion $\alpha$
is as for $\omega$-word acceptors, yielding  B\"{u}chi,
co-B\"{u}chi and parity $\omega$-tree acceptors.
If the transition function specifies at most one permissible
$d$-tuple of states for every element of $Q \times \Sigma^d$, then
the acceptor is deterministic.
The corresponding classes of $\omega$-tree languages are
also denoted by three-letter acronyms, where the last letter
is \class{T} for $\omega$-trees.
% Thus, $\nmt$ is the class of $\omega$-tree languages recognized by nondeterministic Muller tree acceptors.
For $\omega$-trees, the class of all regular $\omega$-tree languages
is $\npt$ and
% Rabin~\cite{Rabin1970} proved that
$\nbt$ is a proper
subclass of $\npt$.
For any automaton or acceptor $M$,
we denote the number of its states by $|M|$.

% in contrast to the situation for $\omega$-word languages.

\section{Derived \texorpdfstring{$\omega$}{omega}-tree languages}%
\label{section-Derived-omega-tree-languages}

Given an $\omega$-tree $t$ we define
the $\omega$-word language $\paths(t)$ consisting of the
$\omega$-words labeling its infinite paths.
That is, we define
\[\paths(t) = \{t(\pi) \mid \pi \textrm{ is an infinite path in } T_d\}.\]
If $L$ is an $\omega$-word language and $d$ is a positive
integer, we define a corresponding
language of $d$-ary $\omega$-trees \concept{derived from} $L$ as follows:
\[\Trees_d(L) = \{t \in T_d^{\Sigma} \mid \paths(t) \subseteq L\}.\]
That is, $\Trees_d(L)$ consists of all $d$-ary $\omega$-trees such
that every infinite path in the tree is labeled by an element of $L$.
If $\class{C}$ is any class of $\omega$-word languages,
$\Trees_d(\class{C})$ denotes the class of all $\omega$-tree
languages $\Trees_d(L)$ such that $L \in \class{C}$.

\subsection{Derived tree languages}%
\label{subsection-derived-tree-languages}

Not every regular $d$-ary $\omega$-tree language
can be derived in this way from an $\omega$-word language.
As an example, consider the language $L_a$ of all binary
$\omega$-trees $t$ over $\Sigma = \{a,b\}$
such that there is at least one node labeled with $a$.
An NBT acceptor can recognize $L_a$ by guessing
and checking a path that leads to an $a$.
However, if $L_a = \Trees_2(L)$ for some $\omega$-word language $L$,
then because there are $\omega$-trees in $L_a$ that have
infinite paths labeled exclusively with $b$,
we must have $b^{\omega} \in L$, so
the binary $\omega$-tree labeled exclusively with $b$ would also be
in $\Trees_2(L)$, a contradiction.

Given an $\omega$-word acceptor $M = (Q,q_0,\delta,\alpha)$,
we may construct a related
$\omega$-tree acceptor $\mtd = (Q,q_0,\delta^{T,d},\alpha)$
as follows.
For all $q \in Q$ and
all $(\sigma_1, \ldots, \sigma_d)  \in \Sigma^d$, define
\[\delta^{T,d}(q,(\sigma_1, \ldots, \sigma_d)) =
\{(q_1, \ldots, q_d) \mid \forall i \in D, q_i \in \delta(q,\sigma_i)\}.\]
That is, the acceptor $\mtd$ may continue the computation at
a child of a node with any state permitted by $M$,
independently chosen.
It is tempting to think that $\lang{\mtd} = \Trees_d(\lang{M})$, but this
may not be true when $M$ is not deterministic.
\begin{lem}%
	\label{lemma-word-to-tree}
	Given an $\omega$-word acceptor $M$, we have that $\lang{\mtd} \subseteq \Trees_d(\lang{M})$
	with equality if $M$ is deterministic.
\end{lem}

\begin{proof}
	Consider the $\omega$-word acceptor $M = (Q,q_0,\delta,\alpha)$.
	If $t \in \lang{\mtd}$ then there is a run $r$ of $\mtd$ on $t$ satisfying
	the acceptance criterion $\alpha$ on every infinite path.
	Thus, $t(\pi) \in \lang{M}$ for every infinite path $\pi$ and
	$t \in \Trees_d(\lang{M})$.

	Suppose $t \in \Trees_d(\lang{M})$ and $M$ is deterministic.
	Then $\mtd$ is also deterministic, and there is a unique
	run $r$ of $\mtd$ on $t$.
	For every infinite path $\pi$,
	$r(\pi)$ is also the unique run of $M$ on the $\omega$-word
	$t(\pi)$, which satisfies $\alpha$ because $t \in \Trees_d(\lang{M})$.
	Thus $t \in \lang{\mtd}$.
\end{proof}

Boker et al.~\cite{BKKS13} give the following example
to show that the containment asserted in Lemma~\ref{lemma-word-to-tree}
may be proper if $M$ is not deterministic.
The $\omega$-language $L$ specified by ${(a+b)}^*b^{\omega}$ can
be recognized by the nondeterministic B\"{u}chi acceptor $M$ with
two states, $q_0$ and $q_1$,
transition function
$\delta(q_0,a) = \{q_0\}$, $\delta(q_0,b) = \{q_0, q_1\}$,
$\delta(q_1,b) = \{q_1\}$, and accepting state set $\{q_1\}$.
Let $d = 2$, specifying binary trees with directions $\{1,2\}$.
Then $M^{T,2}$ is a nondeterministic $\omega$-tree acceptor, but
the following example shows $\lang{M^{T,2}} \subsetneq \Trees_2(L)$.
Consider the binary $\omega$-tree $t$ that labels every node in $1^*2$ with
$a$ and every other non-root node with $b$.
Clearly $t \in \Trees_2(L)$ because every infinite path in $t$
has at most one $a$, but no run of $M^{T,2}$ can satisfy the acceptance
criterion on the path $1^{\omega}$. Suppose $r$ were an accepting run of $M^{T,2}$ on $t$. Then for some $n\geq 0$, $r(1^n)$ would have to be equal to $q_1$. But then such a mapping $r$ would not be a valid run because $1^n2$ is labeled by $a$ and $\delta^{T,2}(q_1,(b,a)) = \emptyset$ because $\delta(q_1,a) = \emptyset$.
% However, $M^{T,2}$ does not accept $t$.  There is no run of $M^{T,2}$ that assigns $q_1$ to any node $1^n$ of the path $1^{\omega}$, because otherwise there is no state that could be assigned to the child node $1^n2$.  Thus, no run of $M^{T,2}$ on $t$ satisfies the acceptance criterion on the path $1^{\omega}$, and $M^{T,2}$ does not accept $t$.

\subsection{Good for trees}%
\label{subsection-Good-for-trees}

This phenomenon motivates the following definition.
An $\omega$-word acceptor $M$ is \concept{good for trees}
iff for any positive integer $d$,
$\lang{\mtd} = \Trees_d(\lang{M})$.
Nondeterministic $\omega$-word acceptors that are good for trees
are equivalent in expressive power to deterministic $\omega$-word
acceptors, as
stated by the following result of Boker et al.

\begin{thmC}[\cite{BKKS13}]%
	\label{theorem-bkks}
	Let $L$ be a regular $\omega$-word language and $d \ge 2$.
	If $\Trees_d(L)$ is recognized
	by a nondeterministic $\omega$-tree acceptor
	with acceptance criterion $\alpha$, then $L$ can be
	recognized by a deterministic $\omega$-word acceptor
	with acceptance criterion $\alpha$.
\end{thmC}

This theorem generalizes prior results of
Kupferman, Safra and Vardi
for B\"{u}chi acceptors~\cite{KSV06}
and
Niwi\'{n}ski and Walukiewicz
for parity acceptors~\cite{NW1998}.
One consequence of Theorem~\ref{theorem-bkks} is that when $d \ge 2$, nondeterministic $\omega$-word acceptors that are good for trees are not more expressive than the corresponding deterministic $\omega$-word acceptors.
% To see this, if $N$ is a nondeterministic $\omega$-word acceptor with acceptance criterion $\alpha$ that is good for trees, then $N^{T,d}$ is a nondeterministic $\omega$-tree acceptor with acceptance criterion $\alpha$ recognizing $\Trees_d(\lang{N})$, so there is a deterministic $\omega$-word acceptor $M$ with acceptance criterion $\alpha$ such that $\lang{M} = \lang{N}$.
%
Also, for $d \ge 2$, nondeterminism does not increase expressive power
over determinism when recognizing $\omega$-tree languages of the form
$\Trees_d(L)$.
To see this, if $N$ is a nondeterministic $\omega$-tree
acceptor with acceptance criterion $\alpha$ recognizing
$\Trees_d(L)$ then there is a deterministic $\omega$-word acceptor $M$
with acceptance criterion $\alpha$ such that $\lang{M} = L$,
and $\mtd$ is a deterministic $\omega$-tree acceptor with acceptance
criterion $\alpha$ that also recognizes $\Trees_d(L)$.

% For languages of the form $\Trees_d(L)$, nondeterminism of tree acceptors does not increase expressive power over determinism when $d \ge 2$.
However, it is possible that nondeterminism permits
acceptors with smaller numbers of states.
Kuperberg and Skrzypczak~\cite{KS15}
have shown that for an NBT acceptor $M$
recognizing the $\omega$-tree language $\Trees_d(L)$, there is a DBW acceptor with at most
$|M|^2$ states recognizing $L$,
so nondeterminism gives at most a quadratic savings
for B\"{u}chi tree acceptors that are good for trees.
However, they have also shown that the blowup in the
case of nondeterministic co-B\"{u}chi tree acceptors
(and all higher parity conditions) is necessarily exponential
in the worst case.

\section{Learning tree languages}%
\label{section-Learning-tree-languages}

We address the problem of learning derived $\omega$-tree languages by giving a polynomial time reduction of the problem of learning $\Trees_d(\class{C})$ to the problem of learning $\class{C}$.
The paradigm of learning we consider is exact learning
with membership queries and equivalence queries.
Maler and Pnueli~\cite{Maler1995} have given a polynomial
time algorithm to learn
the class of weak regular $\omega$-languages
using membership and equivalence queries.
Their algorithm and
the reduction we give in Theorem~\ref{theorem-reduction-general}
prove the following theorem.

\begin{thm}%
	\label{theorem-learn-trees}
	For every positive integer $d$,
	there is a polynomial time algorithm to learn $\Trees_d(\dwpw)$
	using membership and equivalence queries.
\end{thm}

\subsection{Representing examples}%
\label{subsection-Representing-examples}

For a learning algorithm, the examples tested by membership queries
and the counterexamples returned by equivalence queries
need to be finitely represented.
For learning regular $\omega$-word languages, it suffices to consider
\concept{ultimately periodic $\omega$-words}, that is,
words of the form ${u(v)}^{\omega}$ for finite words $u \in \Sigma^*$ and
$v \in \Sigma^+$.
If two regular $\omega$-word languages agree on all the ultimately
periodic $\omega$-words, then they are equal.
The pair $(u,v)$ of finite words represents the ultimately
periodic word ${u(v)}^{\omega}$.

The corresponding class of examples in the case of regular $\omega$-tree
languages is the class of \concept{regular $\omega$-trees}.
These are $\omega$-trees that have a finite number of
nonisomorphic complete infinite subtrees.
We represent a regular $\omega$-tree $t$ by a
\concept{regular $\omega$-tree automaton} $A_t = (Q,q_0,\delta,\tau)$,
where $(Q,q_0,\delta)$ is a complete deterministic finite state
word automaton over the input alphabet $D = \{1,\ldots,d\}$
and $\tau$ is an output function that labels each transition
with an element of $\Sigma$.
That is, $\tau: Q \times D \rightarrow \Sigma$.
The regular $\omega$-tree $t$ represented by such an automaton $A_t$
is defined as follows.
For $x \in D^+$, let $i \in D$ be the last symbol of $x$ and let
$x'$ be the rest of $x$, so that $x = x' \cdot i$.
Then define $t(x) = \tau(\delta(q_0,x'),i)$, that is,
$t(x)$ is the label assigned by $\tau$ to the last transition
in the unique run of $A_t$ on $x$.
%
% For example, for $\Sigma = \{a,b\}$, the binary $\omega$-tree that labels every node in $1^*2$ with $a$ and every other non-root node with $b$ can be represented by the regular $\omega$-tree automaton with three states $q_0$, $q_1$, $q_2$, transition function $\delta(q_0,1) = q_0$, $\delta(q_0,2) = q_1$ and all other transitions to $q_2$, with $\tau$ defined to be $a$ on the transition $(q_0,2)$ and $b$ on all other transitions.

Rabin~\cite{Rabin1972} proved that if two regular $\omega$-tree languages agree on all
the regular $\omega$-trees then they are equal.
Thus, ultimately periodic $\omega$-words
and regular $\omega$-trees are proper subsets of
examples that are nonetheless sufficient
to determine the behavior of regular $\omega$-word and $\omega$-tree acceptors
on all $\omega$-words and $\omega$-trees, respectively.

\subsection{Types of queries for learning}%
\label{subsection-Types-of-queries-for-learning}

We consider the situation in which a learning algorithm
$\A$ is attempting to learn an initially unknown target language
$L$ of $\omega$-words from a known class $\class{C} \subseteq \dbw$.
The information that $\A$ gets about $L$ is in the form
of answers to queries of specific types~\cite{Angluin:1988}.
The learning algorithm will use membership and equivalence
queries, whereas, restricted and unrestricted subset
queries will in addition be considered in the proof.

In a \concept{membership query about} $L$, abbreviated $\MQ$,
the algorithm $\A$
specifies an example as a pair of finite words $(u,v)$
and receives the answer ``yes'' if ${u(v)}^{\omega} \in L$
and ``no'' otherwise.
In an \concept{equivalence query about} $L$, abbreviated $\EQ$,
the algorithm $\A$ specifies a hypothesis language $\lang{M}$
as a DBW acceptor $M$, and receives either
the answer ``yes'' if $L = \lang{M}$, or ``no'' and
a counterexample, that is, a pair of finite words $(u,v)$
such that ${u(v)}^{\omega} \in (L \oplus \lang{M})$,
where $B \oplus C$ denotes the symmetric difference of sets
$B$ and $C$.

In a \concept{restricted subset query about} $L$, abbreviated $\RSQ$,
the algorithm $\A$ specifies a hypothesis language $\lang{M}$
as a DBW acceptor $M$, and receives the
answer ``yes'' if $\lang{M} \subseteq L$ and ``no'' otherwise.
An \concept{unrestricted subset query about} $L$, abbreviated $\USQ$,
is like a restricted subset query,
except that in addition to the answer of ``no'',
a counterexample $(u,v)$ is provided such that
${u(v)}^{\omega} \in (\lang{M} \setminus L)$.

A learning algorithm $\A$ using specific types of queries
\concept{exactly learns} a class $\class{C}$ of $\omega$-word languages
iff
for every $L \in \class{C}$, the algorithm makes a finite number
of queries of the specified types about $L$
and eventually
halts and outputs a DBW acceptor $M$ such that $\lang{M} = L$.
The algorithm runs in polynomial time iff
there is a fixed polynomial $p$ such that for every
$L \in \class{C}$, at every point the number of steps used
by $\A$ is bounded by $p(n,m)$, where $n$ is the size of the
smallest DBW acceptor recognizing $L$, and $m$ is
the maximum length of any counterexample $\A$ has received
up to that point.

The case of a learning algorithm for $\omega$-tree languages is analogous,
except that the examples and counterexamples are given by
regular $\omega$-tree automata,
and the hypotheses provided to equivalence or subset queries
are represented by DBT acceptors.
We also consider cases in which the inputs to equivalence
or subset queries may be NBW or NBT acceptors.

\section{Framework of a reduction}%
\label{section-Framework-of-a-reduction}

Suppose $\A$ is a learning algorithm that uses membership and
equivalence queries and exactly learns a class $\class{C} \subseteq \dbw$.
We shall describe an algorithm $\ATrees$ that uses membership and
equivalence queries and
exactly learns the derived class $\Trees_d(\class{C})$ of
$\omega$-tree languages.
Note that $\Trees_d(\class{C}) \subseteq \dbt$.

The algorithm $\ATrees$ with target concept $\Trees_d(L)$
simulates algorithm $\A$ with target concept $L$.
In order to do so, $\ATrees$ must correctly answer
membership and equivalence queries from $\A$ about $L$
by making one or more membership and/or equivalence queries
of its own about $\Trees_d(L)$.
Before describing the algorithm $\ATrees$
we establish some basic results about regular $\omega$-trees.

\subsection{Testing acceptance of a regular \texorpdfstring{$\omega$}{omega}-tree}%
\label{section-Testing-acceptance-of-t}
\noindent
We describe a polynomial time algorithm $\Acc(A_t,M)$
that takes as input a regular $\omega$-tree $t$
represented by a regular $\omega$-tree automaton
$A_t =(Q_1, q_{0,1}, \delta_1, \tau_1)$
and
a DBW acceptor
$M = (Q_2, q_{0,2}, \delta_2, F_2)$
and determines whether or not $\mtd$ accepts $t$.
If not, it also outputs a pair $(u,v)$ of finite
words such that
${u(v)}^{\omega} \in (\paths(t) \setminus \lang{M})$.

    \begin{algorithm}%
    \scriptsize
	\caption{: $\Acc(A_t,M)$}%
	\label{algorithm-Acc}
	\begin{algorithmic}
		\REQUIRE {$A_t = (Q_1,q_{0,1},\delta_1,\tau_1)$ representing $t$;\\
			$M = (Q_2, q_{0,2}, \delta_2, F_2)$, a complete DBW acceptor}

		\ENSURE {Return ``yes'' if $\mtd$ accepts $t$\\
			else return ``no'' and $(u,v)$ with ${u(v)}^{\omega} \in (\paths(t) \setminus \lang{M})$.}

		\vspace{0.2cm}
		\STATE let $Q = Q_1 \times Q_2$
		\STATE let $q_0 = (q_{0,1},q_{0,2})$
		\FORALL {$(q_1, q_2) \in Q$ and $i \in D$}
		\STATE let $\delta((q_1,q_2),i) =
		(\delta_1(q_1,i),\delta_2(q_2,\tau_1(q_1,i)))$
		\ENDFOR
		\STATE let $F = \{(q_1,q_2) \mid q_2 \in F_2\}$
		\STATE let $M' = (Q,q_0,\delta,F)$
		\IF {$\lang{M'} = D^{\omega}$}
		\RETURN ``yes''
		\ELSE
		\STATE find $x(y)^{\omega} \in (D^{\omega} \setminus \lang{M'})$
		\STATE let ${u(v)}^{\omega} = t(x(y)^{\omega})$
		\RETURN ``no'' and $(u,v)$
		\ENDIF
	\end{algorithmic}
\end{algorithm}

We may assume $M$ is complete by adding (if necessary)
a new non-accepting sink state and directing
all undefined transitions to the new state.
We construct a DBW acceptor $M'$
over the alphabet $D = \{1,\ldots,d\}$
by combining
$A_t$ and $M$ as follows.
The states are $Q = Q_1 \times Q_2$,
the initial state is $q_0 = (q_{0,1},q_{0,2})$,
the set of accepting states is $F = \{(q_1,q_2) \mid q_2 \in F_2\}$,
and the transition function $\delta$ is defined by
$\delta((q_1,q_2),i) = (\delta_1(q_1,i), \delta_2(q_2,\tau_1(q_1,i)))$
for all $(q_1,q_2) \in Q$ and $i \in D$.
For each transition, the output of the regular $\omega$-tree automaton
$A_t$ is the input of the DBW acceptor $M$.

An infinite path $\pi$ in $t$ corresponds to an $\omega$-word
$z \in D^{\omega}$, giving the sequence of directions from the root.
The unique run of $M'$ on $z$ traverses a sequence of states;
if we project out the first component, we get the run
of $A_t$ on $z$,
and if we project out the second component,
we get the run of $M$ on $t(\pi)$.
Then $\mtd$ accepts $t$ iff $M$ accepts $t(\pi)$
for every infinite path $\pi$, which is true iff
$\lang{M'} = D^{\omega}$.
This in turn is true
iff every nonempty accessible recurrent set of
states in $M'$ contains at least one element of $F$.

A set $S$ of states is \concept{recurrent} iff
for all $q, q' \in S$, there is a nonempty finite word $v$
such that $\delta(q,v) = q'$ and for every prefix $u$ of $v$ we have $\delta(q,u)\in S$.
A set $S$ of states is \concept{accessible} iff
for every $q \in S$ there exists a finite word $u$
such that $\delta(q_0,u) = q$.

The algorithm to test whether $\mtd$ accepts $t$
first removes from the transition graph of $M'$ all states that are not
accessible.
It then removes all states in $F$ and tests whether there
is any cycle in the remaining graph.
If not, then $\mtd$ accepts $t$.
Otherwise, there is a state $q$ in $Q$ and
finite words $x \in D^*$ and $y \in D^+$ such that
$\delta(q_0,x) = q$ and $\delta(q,y) = q$ and none
of the states traversed from $q$ to $q$ along the path $y$
are in $F$.
Thus, ${x(y)}^{\omega}$ is an ultimately periodic path $\pi$
that does not visit $F$ infinitely often,
and letting ${u(v)}^{\omega}$ be $t({x(y)}^{\omega})$, we
have ${u(v)}^{\omega} \in (\paths(t) \setminus \lang{M})$,
so the pair $(u,v)$ is returned in this case.
The required graph operations are standard and
can be accomplished in time polynomial in
$|M|$ and $|A_t|$.

\subsection{Representing a language as paths of a tree}%
\label{subsection-Representing-a-language-as-paths-of-a-tree}

When the algorithm $\ATrees$ makes a membership query
about $\Trees_d(L)$ with a regular $\omega$-tree $t$,
the answer is ``yes'' if $\paths(t) \subseteq L$ and
``no'' otherwise.
Thus, this query has the effect of a restricted subset
query about $L$ with $\paths(t)$.
However, this does not give us restricted subset queries
for arbitrary $\dbw$ languages.
Next, we examine the close relationship between
languages of the form $\paths(t)$ and safety languages.

An $\omega$-word language $L$ is a \concept{safety language}
iff $L$ is a regular $\omega$-word language and
for every $\omega$-word $w$ not in $L$, there exists
a finite prefix $x$ of $w$ such that no $\omega$-word with
prefix $x$ is in $L$. A language is safety iff it is in the class $\dwcw$.
An alternative characterization is that there is an
NBW acceptor $M = (Q,q_0,\delta,Q)$, all of whose
states are accepting, such that $\lang{M} = L$.
In this case, the acceptor is typically not complete
(otherwise it recognizes $\Sigma^{\omega}$).
An example of a language in $\dwpw$ that is not a
safety language is $a^* b^* {(a)}^{\omega}$.
Although $b^{\omega}$ is not in the language, every
finite prefix $b^k$ is a prefix of some
$\omega$-word in the language.

\begin{lem}%
	\label{lemma-tree-nbw}
	If $A_t$ is a regular $\omega$-tree automaton representing an $\omega$-tree $t$,
	then $\paths(t)$ is a safety language recognizable
	by an NBW acceptor $M$ with $|M| = |A_t|$.
\end{lem}

\begin{proof}
	If $A_t = (Q,q_0,\delta,\tau)$, then we define $M = (Q,q_0,\delta',Q)$
	where
	\[\delta'(q,\sigma) = \{r \in Q \mid (\exists i \in D) (\delta(q,i) = r \wedge \tau(q,i) = \sigma)\}\]
	for all $q \in Q$ and $\sigma \in \Sigma$.
	That is, the $M$ transition on $q$ and $\sigma$ is defined to be all states reachable from $q$
	by a transition in $A_t$ labeled with $\sigma$.
	Note that all states of $M$ are accepting.

	If $w \in \paths(t)$, then there is a run $r_0, r_1, \ldots$ of $A_t$
	whose transitions are labeled by $w$, and this is a run of $M$ on $w$, so $w \in \lang{M}$.
	Conversely, if $w \in \lang{M}$, then there is some run $r_0, r_1, \ldots$ of $M$ on $w$,
	and this is a run of $A_t$ whose transitions are labeled with $w$,
	and therefore $w \in \paths(t)$.
\end{proof}

For the converse, representing a safety language as the paths of a regular $\omega$-tree,
we require a lower bound on $d$, the arity of the tree.
If $M = (Q,q_0,\delta,F)$ is an NBW acceptor
and $q \in Q$, we define the set of transitions out of $q$ to be
$\transitions(q) = \{(\sigma,r) \mid \sigma \in \Sigma \wedge r \in \delta(q,\sigma)\}$.
We define the \concept{out-degree} of $M$
to be the maximum over $q \in Q$ of the cardinality of $\transitions(q)$.

\begin{lem}%
	\label{lemma-nbw-tree}
	Let $L$ be a safety language recognized by NBW acceptor
	$M = (Q, q_0, \delta, Q)$.
	Suppose the out-degree of $M$ is at most $d$.
	Then there is a $d$-ary regular $\omega$-tree $t$
	such that $\paths(t) = L$, and $t$ is representable by $A_t$ with
	$|A_t| = |M|$.
\end{lem}

\begin{proof}
	We may assume that every state of $M$ is accessible and has at least one
	transition defined.
	We define $A_t = (Q,q_0,\delta_t,\tau)$ over the alphabet $D = \{1,\ldots,d\}$
	as follows.
	For $q \in Q$, choose a surjective mapping $f_q$ from $D$ to $\transitions(q)$.
	Then for $q \in Q$ and $i \in D$, let $(\sigma,r) = f_q(i)$ and define
	$\delta_t(q,i) = r$ and $\tau(q,i) = \sigma$.

	If $w \in L$, then there is a run $r_0, r_1, \ldots$ of $M$ on $w$, and
	there is an infinite path in $A_t$ traversing the same states in which
	the labels are precisely $w$, so $w \in \paths(t)$.
	Conversely, if $w \in \paths(t)$, then there is an infinite path
	$\pi$ such that $t(\pi) = w$, and the sequence of states of $A_t$
	traversed by $w$ yields a run of $M$ on $w$, so $w \in L$.
\end{proof}

The NBW acceptor in the proof of Lemma~\ref{lemma-tree-nbw} can be determinized
via the subset construction to give a DBW acceptor of size at most
$2^{|A_t|}$ recognizing the same language.
In the worst case this exponential blow up in converting a regular $\omega$-tree
automaton to a DBW acceptor is necessary, as shown by the following lemma.

% this is just an exponential blowup NBW -> DBW for a safety language
% which is presumably somewhere in the literature

\begin{lem}\label{lemma-exp-blowup-nbw-dbw}
	There exists a family of regular $\omega$-trees $t_1, t_2, \ldots$ such that
	$t_n$ can be represented by a regular $\omega$-tree automaton of size $n+2$,
	but the smallest DBW acceptor recognizing $\paths(t_n)$ has size at
	least $2^n$.
\end{lem}

\begin{proof}
	Let $\Sigma = \{a,b,c\}$ and let $L_n$ be ${(a + b + (a{(a+b)}^n c))}^{\omega}$.
	This is a safety language: $w \in L_n$ iff every occurrence of $c$ in $w$
	is preceded by a word of the form $a{(a+b)}^n$.
	There is a NBW acceptor $M_n$ of $n+2$ states recognizing $L_n$.
	The states are nonnegative integers in $[0,n+1]$, with $0$ the initial
	state,
	$\delta(0,a) = \{0,1\}$, $\delta(0,b) = 0$,
	$\delta(i,a) = \delta(i,b) = i+1$ for $1 \le i \le n$,
	and $\delta(n+1,c) = 0$.

	By Lemma~\ref{lemma-nbw-tree}, there is a ternary regular $\omega$-tree
	$t_n$ such that $\paths(t_n) = L_n$ and $t_n$ is represented by a
	regular $\omega$-tree automaton with $n+2$ states.
	However, any
	DBW acceptor recognizing $L_n$ must have enough states to distinguish
	all $2^n$ strings in ${(a+b)}^n$ in order to check the safety condition.
\end{proof}

If $t$ is a $d$-ary regular $\omega$-tree represented by
the regular $\omega$-tree automaton $A_t$, then
$\acceptor(A_t)$ denotes the NBW acceptor $M$ recognizing $\paths(t)$
constructed from $A_t$ in the proof of Lemma~\ref{lemma-tree-nbw}.
Note that the out-degree of $\acceptor(A_t)$ is at most $d$.

If $M$ is an NBW acceptor such that $\lang{M}$ is a safety language
and the out-degree of $M$ is at most $d$, then
$\tree_d(M)$ denotes the regular $\omega$-tree automaton $A_t$
constructed from $M$ in the proof of Lemma~\ref{lemma-nbw-tree}.
We also use the notation $\tree_d(L)$ if $L$ is a safety language
and the implied acceptor for $L$ is clear.

For example, given finite words $u \in \Sigma^*$ and $v \in \Sigma^+$,
the singleton set containing ${u(v)}^{\omega}$ is a safety language
recognized by a DBW of out-degree $1$ and size linear in $|u|+|v|$.
Then $\tree_d({u(v)}^{\omega})$ represents the $d$-ary tree all of whose
infinite paths are labeled with ${u(v)}^{\omega}$.

\section{The algorithm \texorpdfstring{$\ATrees$}{A-trees}}%
\label{section-The-algorithm-ATrees}

We now describe the algorithm $\ATrees$,
which learns $\Trees_d(L)$
by simulating the algorithm $\A$
and answering the membership and equivalence queries of $\A$ about $L$.
It is summarized in Algorithm~\ref{algorithm-ATrees}, and some of the cases are illustrated in an example presented in Appendix~\ref{app:A-tree-example}.
\begin{algorithm}[t]%
    \footnotesize
	%%%% algorithm A_{\Trees} simulating A
	\caption{: $\ATrees$}%
	\label{algorithm-ATrees}
	\begin{algorithmic}[t]
		\REQUIRE {
			Learning algorithm $\A$ for $\class{C}$;\\
			$\MQ$ and $\EQ$ access to $\Trees_d(L)$ for $L \in \class{C}$}

		\ENSURE {Acceptor $\mtd$ such that $\lang{\mtd} = \Trees_d(L)$}

		\vspace{0.2cm}
		\WHILE{$\A$ has not halted}
		\IF {next step of $\A$ is not a query}
		\STATE {simulate next step of $\A$}
		\ELSIF {$\A$ asks $\MQ(u,v)$ about $L$}
		\STATE {answer $\A$ with $\MQ(\tree_d({u(v)}^{\omega}))$ about $\Trees_d(L)$}
		\ELSIF {$\A$ asks $\EQ(M)$ about $L$}
		\STATE {ask $\EQ(\mtd)$ about $\Trees_d(L)$}
		\IF {$\EQ(\mtd)$ answer is ``yes''}
		\RETURN {$\mtd$ and halt}
		\ELSE [$\EQ(\mtd)$ answer is counterexample tree $t$ given by $A_t$]
		\IF {$\Acc(A_t,M)$ returns ``no'' with value $(u,v)$}
		\STATE {answer $\A$ with $(u,v)$}
		\ELSE [$\Acc(A_t,M)$ returns ``yes'']
		\STATE {let $M' = \acceptor(A_t)$}
		\FORALL{accepting states $q$ of $M'$}
		\STATE{simulate in parallel $\Findctrex(M',q)$}
		\STATE{terminate all computations and answer $\A$ with the first $(u,v)$ returned}
		\ENDFOR
		\ENDIF
		\ENDIF
		\ENDIF
		\ENDWHILE [$\A$ halts with output $M$]
		\RETURN {$\mtd$ and halt}
	\end{algorithmic}
\end{algorithm}

If $\A$ asks a membership query with $(u,v)$ then
$\ATrees$ constructs the regular $\omega$-tree automaton
$\tree_d({u(v)}^{\omega})$ representing the $d$-ary regular $\omega$-tree
all of whose infinite paths are labeled ${u(v)}^{\omega}$,
and makes a membership query with $\tree_d({u(v)}^{\omega})$.
Because ${u(v)}^{\omega} \in L$ iff the tree represented by
$\tree_d({u(v)}^{\omega})$ is in $\Trees_d(L)$,
the answer to the query about $\tree_d({u(v)}^{\omega})$ is simply given
to $\A$ as the answer to its membership query about $(u,v)$.

For an equivalence query from $\A$ specified by a DBW acceptor $M$,
the algorithm $\ATrees$ constructs the corresponding DBT
acceptor $\mtd$, which recognizes $\Trees_d(\lang{M})$,
and makes an equivalence query with $\mtd$.
If the answer is ``yes'', the algorithm $\ATrees$ has succeeded
in learning the target $\omega$-tree
language $\Trees_d(L)$
and outputs $\mtd$ and halts.
Otherwise, the counterexample returned is a regular $\omega$-tree $t$
in $\lang{\mtd} \oplus \Trees_d(L)$, represented by a regular
$\omega$-tree automaton $A_t$.
A call to the algorithm $\Acc(A_t,M)$ determines whether $\mtd$ accepts $t$.
If $\mtd$ rejects $t$, then $t \in \Trees_d(L)$ and $t$ is
a \concept{positive counterexample}.
If $\mtd$ accepts $t$, then $t \not\in \Trees_d(L)$ and $t$ is
a \concept{negative counterexample}.
We next consider these two cases.

If $t$ is a positive counterexample then we know that
$t \in \Trees_d(L)$ and therefore
$\paths(t) \subseteq L$.
Because $t \notin \lang{\mtd}$, the acceptor $M$ must reject at least one
infinite path in $t$.
In this case, the algorithm $\Acc(A_t,M)$ returns a pair of finite words
$(u,v)$ such that ${u(v)}^{\omega} \in (\paths(t) \setminus \lang{M})$,
and therefore ${u(v)}^{\omega} \in (L \setminus \lang{M})$.
The algorithm $\ATrees$ returns
the positive counterexample
$(u,v)$ to $\A$ in response to its equivalence query with $M$.

If $t$ is a negative counterexample, that is,
$t \in (\lang{\mtd} \setminus \Trees_d(L))$, then we know that
$\paths(t) \subseteq \lang{M}$, but at least one
element of $\paths(t)$ is not in $L$, so
$(\lang{M} \setminus L) \neq \emptyset$.
Ideally, we would like to extract an ultimately
periodic $\omega$-word ${u(v)}^{\omega} \in (\paths(t) \setminus L)$
and provide $(u,v)$ to $\A$ as a negative counterexample
in response to its equivalence query with $M$.

If we could make an \emph{unrestricted subset query}
with $\paths(t)$ about $L$, then the counterexample
returned would be precisely what we need.

As noted previously, if $t$ is any regular $\omega$-tree
then we can simulate a restricted subset query
with $\paths(t)$ about $L$ by making a membership query
with $t$ about $\Trees_d(L)$, because $\paths(t) \subseteq L$
iff $t \in \Trees_d(L)$.
In order to make use of this, we next show how
to use restricted subset queries about $L$ to implement
an unrestricted subset query about $L$.

\subsection{Restricted subset queries}%
\label{subsection-Restricted-subset-queries}

To establish basic techniques,
we show how to reduce unrestricted subset
queries to restricted subset queries for nondeterministic
or deterministic finite acceptors over finite words.
Suppose $L \subseteq \Sigma^*$ and
we may ask restricted subset queries about $L$.
In such a query,
the input is a nondeterministic (resp., deterministic)
finite acceptor $M$,
and the answer is ``yes'' if $\lang{M}$ is a subset of $L$,
and ``no'' otherwise.
If the answer is ``no'', we show how to find
a shortest counterexample $u \in (\lang{M} \setminus L)$
in time polynomial in $|M|$ and $|u|$.

\begin{thm}%
	\label{thm-nfa-dfa-subset-reduction}
	There is an algorithm $\R^*$ which takes as input an NFW (resp., DFW) $M$,
	and has restricted subset query access to a language $L$
	with NFW (resp., DFW) acceptors
	as inputs,
	that correctly answers the unrestricted subset query with
	$M$ about $L$.
	Additionally, if $L$ is recognized by a DFW $T_L$, then
	$\R^*(M)$ runs in time bounded by a polynomial in $|M|$ and $|T_L|$.\footnote{The cardinality of $\Sigma$ is treated as a constant.}
\end{thm}

The idea of the proof is to first establish the minimal length $\ell$ of a counterexample,
and then try to extend the prefix $\epsilon$ letter by letter until obtaining a full length minimal counterexample. Note that trying to establish a prefix of a counterexample letter by letter, without obtaining a bound first, may not terminate. For instance, if $L=\Sigma^* \setminus a^*b$, one can establish the sequence of prefixes $\epsilon, a, aa, aaa, \ldots$ and never reach a counterexample.

To prove Theorem~\ref{thm-nfa-dfa-subset-reduction} we first construct an acceptor $M_{\ell,v}$ for
$\lang{M}[\ell,v]$, the length and prefix restricted
version of $\lang{M}$, given $M$, $\ell$ and $v$
as inputs.

\begin{lem}%
	\label{lemma-length-and-prefix-restriction}
	There is a polynomial time algorithm to construct
	an acceptor $M_{\ell,v}$ for $\lang{M}[\ell,v]$ given a NFW acceptor $M$, a nonnegative integer $\ell$ and
	a finite word $v$, such that
	\begin{enumerate}
		\item $M_{\ell,v}$ has at most one accepting state, which has
		no out-transitions,
		\item the out-degree of $M_{\ell,v}$ is at most the out-degree of $M$,
		\item $M_{\ell,v}$ is deterministic if $M$ is deterministic.
	\end{enumerate}
\end{lem}

\commentout{
\begin{lem}{\ref{lemma-length-and-prefix-restriction}}[restated]
	There is a polynomial time algorithm to construct
	an acceptor $M_{\ell,v}$ for $\lang{M}[\ell,v]$ given a nondeterministic
	finite acceptor $M$, a nonnegative integer $\ell$ and
	a finite word $v$, such that
	\begin{enumerate}
		\item $M_{\ell,v}$ has at most one accepting state, which has
		no out-transitions,
		\item the out-degree of $M_{\ell,v}$ is at most the out-degree of $M$,
		\item $M_{\ell,v}$ is deterministic if $M$ is deterministic.
	\end{enumerate}
\end{lem}}

\begin{proof}
	If $\ell < |v|$, then $\lang{M}[\ell,v] = \emptyset$, and the output
	$M_{\ell,v}$ is a one-state acceptor with no accepting states.
	Otherwise, assume $v = \sigma_1 \sigma_2 \cdots \sigma_k$
	and construct $M'$ to be the deterministic finite
	acceptor for $v \cdot \Sigma^{\ell - |v|}$ with states $0, 1, \ldots, \ell$
	where $0$ is the inital state, $\ell$ is the final state, and
	the transitions are $\delta(i,\sigma_{i+1}) = i+1$ for $0 \le i < k$
	and $\delta(i,\sigma) = i+1$ for $k \le i < \ell$ and $\sigma \in \Sigma$.

	Then $M_{\ell,v}$ is obtained by a standard product construction of $M$
	and $M'$ for the intersection $\lang{M} \cap \lang{M'}$, with the
	observation that no accepting state in the product has any
	out-transitions defined, so they may all be identified.
	It is straightforward to verify the required properties of $M_{\ell,v}$.
\end{proof}

\begin{proof}[Proof of Theorem~\ref{thm-nfa-dfa-subset-reduction}]
	For input $M$, define $M_{[\ell,v]}$ to be the finite
	acceptor constructed by the algorithm of
	Lemma~\ref{lemma-length-and-prefix-restriction} to recognize
	the length and prefix restricted language $\lang{M}[\ell,v]$.

	For $\ell = 0, 1, 2, \ldots$,
	ask a restricted subset query with $M_{[\ell,\varepsilon]}$,
	until the first query answered ``no''.
	At this point, $\ell$ is the shortest length
	of a counterexample in $(\lang{M} \setminus L)$.
	Then a counterexample $u$ of length $\ell$ is constructed
	symbol by symbol.

	Assume we have found a prefix $u'$ of a
	counterexample of length $\ell$ in $(\lang{M} \setminus L)$,
	with $|u'| < \ell$.
	For each symbol $\sigma \in \Sigma$
	we ask a restricted subset query with $M_{[\ell,u' \sigma]}$,
	until the first query answered ``no''.
	At this point, $u'$ is extended to $u' \sigma$.
	If the length of $u' \sigma$ is now $\ell$, then $u = u' \sigma$
	is the desired counterexample; otherwise, we
	continue extending $u'$.

	Note that if the input $M$ is deterministic, then all
	of the restricted subset queries are made with deterministic
	finite acceptors.
	If $L$ is recognized by a deterministic finite acceptor $T_L$, then
	the value of $\ell$ is bounded by $|M| \cdot |T_L|$,
	and the algorithm runs in time bounded by a polynomial in
	$|M|$ and $|T_L|$.
\end{proof}

We now turn to the $\omega$-word case.

\begin{thm}%
	\label{theorem-restricted-subset-nbw-dbw}
	There is an algorithm $\R^{\omega}$ with input $M$
	and restricted subset query access about $L$,
	(a language recognized by a DBW acceptor $T_L$)
	that correctly answers the unrestricted subset query with $M$ about $L$.
	The algorithm
	$\R^{\omega}(M)$ runs in time bounded by a polynomial in $|M|$ and $|T_L|$.
	If $M$ is a DBW acceptor, then all the restricted subset
	queries will also be with DBW acceptors.
\end{thm}

%%% procedure \R^{\omega} to answer USQ(M)

\begin{algorithm}[h]%
    \footnotesize
	\caption{: $\R^{\omega}(M)$, implementing $\USQ(M)$}%
	\label{algorithm-Romega}
	\begin{algorithmic}
		\REQUIRE {$\RSQ$ access to $L$;\\
			$M = (Q,q_0,\delta,F)$, an NBW acceptor}

		\ENSURE {``yes'' if $\lang{M} \subseteq L$, else ``no'' and $(u,v)$ s.t.
			${u(v)}^{\omega} \in (\lang{M} \setminus L)$}

		\vspace{0.2cm}
		\IF {$\RSQ(M) =$ ``yes''}
		\RETURN ``yes''
		\ELSE
		\STATE find $q \in F$ such that $\RSQ(M_q) =$ ``no''
		\RETURN {``no'' and $\Findctrex(M,q)$}
		\ENDIF

	\end{algorithmic}
\end{algorithm}

For the sake of generality, the proof considers subset queries with NBW acceptors.
The procedure $\R^{\omega}(M)$ takes as input an
NBW acceptor $M$, and has restricted subset query access
(with NBW acceptors as inputs)
to $L$; it is summarized in Algorithm~\ref{algorithm-Romega}.
It first asks a restricted subset query with $M$ about $L$,
returning the answer ``yes'' if its query is answered ``yes''.
Otherwise, for each $q \in F$, it constructs the acceptor
$M_q = (Q,q_0,\delta,\{q\})$ with the single accepting
state $q$ and asks a restricted subset query with $M_q$
about $L$, until the first query answered ``no''.
There will be at least one such query answered ``no''
because any element of
$(\lang{M} \setminus L)$ must visit at least one
accepting state $q$ of $M$ infinitely many times,
and will therefore be in $\lang{M_q}$.
The procedure $\R^{\omega}(M)$ then calls the procedure
$\Findctrex(M,q)$ to find a counterexample to return  --- i.e., a pair $(u,v)$ such
that ${u(v)}^{\omega} \in (\lang{M_q} \setminus L)$, and thus also ${u(v)}^{\omega} \in (\lang{M} \setminus L)$.

\subsection{Producing a counterexample}%
\label{subsection-producing-a-counterexample}

The first challenge encountered in producing a counterexample, in comparison to the finite word case, is that one needs to work out both the period and the prefix of the counterexample to be found, and the two are correlated.
%Next is the description of the procedure $\Findctrex(M,q)$.
Define $L_{q_0,q}$ to be the set of finite words
that lead from the initial state $q_0$ to the state
$q$ in $M$, and define $L_{q,q}$ to be the set
of nonempty finite words that lead from $q$ back
to $q$ in $M$.
Because the language $L_{q_0,q} \cdot {(L_{q,q})}^{\omega}$ is exactly
the set of strings recognized by $M_q$, we know that
$L_{q_0,q} \cdot {(L_{q,q})}^{\omega} \setminus L \neq \emptyset$.

The procedure $\Findctrex(M,q)$ first finds a suitable period, corresponding to a bounded size of a prefix yet to be found, and then finds a prefix of that size in a similar manner to the finite word case.
An example is shown in Appendix~\ref{app:findctrex-example}.
%%% procedure $\Findctrex(M,q)$ to return $(u,v)$

\begin{algorithm}%
    \footnotesize
	\caption{: $\Findctrex(M,q)$}%
	\label{algorithm-Findctrex}
	\begin{algorithmic}
		\REQUIRE {$\RSQ$ access to $L$;\\
			$M = (Q,q_0,\delta,F)$, an NBW acceptor;\\
			$q \in F$;\\
			$L_{q_0,q} \cdot (L_{q,q})^{\omega} \setminus L \neq \emptyset$}

		\ENSURE {$(u,v)$ such that
			${u(v)}^{\omega} \in (L_{q_0,q} \cdot (L_{q,q})^{\omega} \setminus L)$}

		\vspace{0.2cm}
		\STATE let $v = \Findperiod(M,q)$
		\STATE let $u = \Findprefix(M,q,v)$
		\RETURN $(u,v)$
	\end{algorithmic}
\end{algorithm}

Since finding the period is more challenging than the prefix, we explain the procedure $\Findprefix(M,q,v)$ first.
The procedure $\Findprefix(M,q,v)$, summarized in Algorithm~\ref{algorithm-Findprefix},
finds a prefix word $u$ given a
period word $v$ which loops on state $q$ and is guaranteed to be a period of a valid counterexample.
It first finds a length $k$ such that there exists $u\in L_{q_0,q}$ of length $k$ such that $uv^\omega \notin L$. Then it finds such a
word $u$ symbol by symbol.
Note that it uses length and prefix restricted versions
of $L_{q_0,q}$.

%%% procedure \Findprefix u, given v

\begin{algorithm}%
    \footnotesize
	\caption{: $\Findprefix(M,q,v)$}%
	\label{algorithm-Findprefix}
	\begin{algorithmic}
		\REQUIRE {$\RSQ$ access to $L$;\\
			$M = (Q,q_0,\delta,F)$, an NBW acceptor;\\
			$q \in F$;\\
			$v \in L_{q,q}$;\\
			$L_{q_0,q} \cdot (v)^{\omega} \setminus L \neq \emptyset$}

		\ENSURE {$u \in L_{q_0,q}$ such that
			${u(v)}^{\omega} \in (L_{q_0,q} \cdot (v)^{\omega} \setminus L)$}

		\vspace{0.2cm}
		\STATE search for nonnegative integer $k$ such that
		$\RSQ(L_{q_0,q}[k] \cdot (v)^{\omega}) =$ ``no''
		\STATE let $u = \varepsilon$
		\WHILE {$|u| < k$}
		\STATE find $\sigma \in \Sigma$ such that
		$\RSQ(L_{q_0,q}[k,u\sigma] \cdot (v)^{\omega}) =$ ``no''
		\STATE set $u = u \cdot \sigma$
		\ENDWHILE
		\RETURN $u$
	\end{algorithmic}
\end{algorithm}

 Finding the periodic part is much more challenging. Indeed, even if one knows that there is a period of the form ${(a\Sigma^\ell)}^\omega$ for some $\ell$ then the size of the smallest period may be bigger than $\ell+1$. For instance, if $L = \Sigma^\omega \setminus {(abbaccadd)}^\omega$ then there is a period of the form ${(a\Sigma^2)}^\omega$ but the shortest period of a counterexample is of size $9$.

 %%% procedure \Findperiod to find $v_i ... v_j$

 \begin{algorithm}[ht]%
    \footnotesize
 	\caption{: $\Findperiod(M,q)$}%
 	\label{algorithm-Findperiod}
 	\begin{algorithmic}
 		\REQUIRE {$\RSQ$ access to $L$;\\
 			$M = (Q, q_0, \delta, F)$, an NBW acceptor;\\
 			$q \in F$;\\
 			$L_{q_0,q} \cdot (L_{q,q})^{\omega} \setminus L \neq \emptyset$}

 		\ENSURE {$v \in L_{q,q}$ such that
 			$L_{q_0,q} \cdot (v)^{\omega} \setminus L \neq \emptyset$}

 		\vspace{0.2cm}
 		\STATE let $y = \varepsilon$
 		\FORALL {integers $n = 1,2,3,\ldots$}
 		\STATE let $v_n = \Nextword(M,q,y)$
 		\STATE set $y = y \cdot v_n$
 		\FOR {integers $i$, $j$ with $1 \le i \le j \le n$}
 		\FOR {$k = 0$ to $n|M|$}
 		\IF {$\RSQ(L_{q_0,q}[k] \cdot (v_i \cdots v_j)^{\omega}) = $ ``no''}
 		\RETURN {$v = v_i \cdots v_j$}
 		\ENDIF
 		\ENDFOR
 		\ENDFOR
 		\ENDFOR
 	\end{algorithmic}
 \end{algorithm}

Procedure $\Findperiod(M,q)$, summarized in Algorithm~\ref{algorithm-Findperiod},  starts from the condition
\[L_{q_0,q} \cdot {(L_{q,q})}^{\omega} \setminus L \neq \emptyset\]
and finds a sequence of words $v_1, v_2, \ldots \in L_{q,q}$ such that for each $n \ge 1$,
\[L_{q_0,q} \cdot {(v_1 v_2 \cdots v_n \cdot L_{q,q})}^{\omega} \setminus L
\neq \emptyset.\]
For a sufficiently long such sequence, there exists a subsequence
$v = (v_i \cdots v_j)$
that is a suitable period word, as we prove in Section~\ref{subsection-length-restrictions-time-bounds}.

The procedure $\Nextword(M,q,y)$, summarized in Algorithm~\ref{algorithm-Nextword},
is called with $y = v_1 v_2 \cdots v_n$
and finds a suitable next word $v_{n+1}$.
After determining a length $\ell$, it repeatedly
calls the procedure $\Nextsymbol(M,q,y,\ell,v')$ to determine
the next symbol of a suitable word of length $\ell$.

%%% procedure \Nextword to find $v'$ given $y = v_1 \cdots v_n$

\hspace{-0.83cm}
\begin{minipage}[b]{0.54\textwidth}
\begin{algorithm}[H]%
    \footnotesize
	\caption{: $\Nextword(M,q,y)$}%
	\label{algorithm-Nextword}
	\begin{algorithmic}
		\REQUIRE {$\RSQ$ access to $L$;\\
			$M = (Q, q_0, \delta, F)$, an NBW acceptor;\\
			$q \in F$;\\
			$y \in L_{q,q}$ or $y = \varepsilon$;\\
			$L_{q_0,q} \cdot (y \cdot L_{q,q})^{\omega} \setminus L \neq \emptyset$}

		\ENSURE {$v' \in L_{q,q}$ such that
			$L_{q_0,q} \cdot (y \cdot v' \cdot L_{q,q})^{\omega} \setminus L \neq \emptyset$}

		\vspace{0.2cm}
		\STATE search for  integers $k,\ell \geq 0$ s.t.\\
		$\RSQ(L_{q_0,q}[k] \cdot (y \cdot L_{q,q}[\ell])^{\omega}) =$ ``no''
		\STATE let $v' = \varepsilon$

		\WHILE {$|v'| < \ell$}
		\STATE let $\sigma = \Nextsymbol(M,q,y,\ell,v')$
		\STATE set $v' = v' \cdot \sigma$
		\ENDWHILE
		\RETURN $v'$
	\end{algorithmic}
\end{algorithm}
\end{minipage}
\hspace{2ex}
\begin{minipage}[b]{0.43\textwidth}
		\begin{figure}[H]
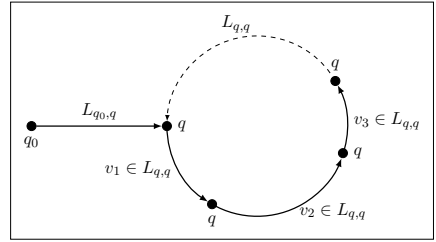

				\caption{An illustration \\ of algorithm $\Nextword$}\label{fig:nextword-illustration}
		\end{figure}
\scalebox{0.6}{
    \begin{tikzpicture}[framed]
			\node [label=below:{$q_0$},circle,fill=black,draw=black,inner sep=1pt, minimum size=0.2cm] (q0) at (0,7) {};
			\node [label=right:{$q$},circle,fill=black,draw=black,inner sep=1pt, minimum size=0.2cm] (q1) at (3,7) {};
			\node [label=below:{$q$},circle,fill=black,draw=black,inner sep=1pt, minimum size=0.2cm] (q2) at (6-2,7-1.73205080757) {};
			\node [label=right:{$q$},circle,fill=black,draw=black,inner sep=1pt, minimum size=0.2cm] (q3) at (6.9,6.4) {};
			\node [label=above:{$q$},circle,fill=black,draw=black,inner sep=1pt, minimum size=0.2cm] (q4) at (6.73205080757,7+1) {};
			\node [circle, fill=none, draw=none] (e) at (8,7) {};

			\draw[-latex, thick] (q0) -- (q1) node [midway, above, fill=none] {$L_{q_0,q}$};
			\draw[-latex,thick,black] ([shift=(180:2cm)]5,7) arc (-180:-123:2cm) node [midway, left, fill=none] {$v_1\in L_{q,q}$};
			\draw[-latex,thick,black] ([shift=(240:2cm)]5,7) arc (-120:-21:2cm) node [midway, right, fill=none] {$\,\,v_2\in L_{q,q}$};
			\draw[-latex,thick,black] ([shift=(-20:2cm)]5,7) arc (-20:27:2cm) node [midway, right, fill=none] {$v_3\in L_{q,q}$};
			\draw[-latex,dashed,black] ([shift=(30:2cm)]5,7) arc (30:177:2cm) node [midway, above, fill=none] {$L_{q,q}$};
		\end{tikzpicture}
}
\end{minipage}

\vspace{10pt}

The procedure $\Nextsymbol(M,q,y,\ell,v')$, summarized in Algorithm~\ref{algorithm-Nextsymbol},
is called to find
a feasible next symbol with which to extend $v'$ in
the procedure $\Nextword$.

%%% procedure \Nextsymbol to find $\sigma$

\hspace{-0.83cm}
\begin{minipage}[b]{0.54\textwidth}
\begin{algorithm}[H]%
    \footnotesize
	\caption{: $\Nextsymbol(M,q,y,\ell,v')$}%
	\label{algorithm-Nextsymbol}
	\begin{algorithmic}
		\REQUIRE {$\RSQ$ access to $L$;\\
			$M = (Q,q_0,\delta,F)$, an NBW acceptor;\\
			$q \in F$;\\
			$y \in L_{q,q}$;\\
			$v' \in \Sigma^*$, $|v'| < \ell$;\\
			$L_{q_0,q} \cdot (y \cdot L_{q,q}[\ell,v'] \cdot L_{q,q})^{\omega}
			\setminus L \neq \emptyset$}

		\ENSURE {$\sigma \in \Sigma$ such that
			$L_{q_0,q} \cdot (y \cdot L_{q,q}[\ell,v'\sigma] \cdot L_{q,q})^{\omega}
			\setminus L \neq \emptyset$}

		\vspace{0.2cm}
		\STATE find integers $k \ge 0$, $m \ge 1$, and $\sigma \in \Sigma$ such that\\
		{$\RSQ(L_{q_0,q}[k] \cdot (y \cdot L_{q,q}[\ell,v'\sigma] \cdot L_{q,q}[m])^{\omega}) = $ ``no''}
		\RETURN {$\sigma$}
	\end{algorithmic}
\end{algorithm}
\end{minipage}
\hspace{2ex}
\begin{minipage}[b]{0.43\textwidth}
	\begin{figure}[H]
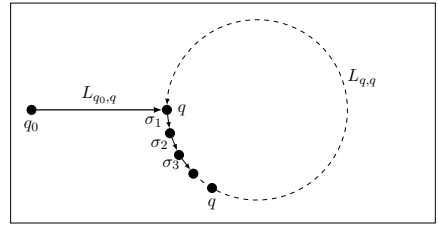

	\caption{An illustration \\ of algorithm $\Nextsymbol$}\label{fig:nextsymbol-illustration}
	\end{figure}
	    \scalebox{0.6}{
		\begin{tikzpicture}[framed]
			\node [label=below:{$q_0$},circle,fill=black,draw=black,inner sep=1pt, minimum size=0.2cm] (q0) at (0,7) {};
			\node [label=right:{$q$},circle,fill=black,draw=black,inner sep=1pt, minimum size=0.2cm] (q1) at (3,7) {};
			\node [label=right:{},circle,fill=black,draw=black,inner sep=1pt, minimum size=0.2cm] (s1) at (6-2.93185165258,7-0.5176380902) {};
			\node [label=right:{},circle,fill=black,draw=black,inner sep=1pt, minimum size=0.2cm] (s2) at (6-2.73205080757,7-1) {};
			\node [label=right:{},circle,fill=black,draw=black,inner sep=1pt, minimum size=0.2cm] (s3) at (6-2.41421356237,7-1.41421356237) {};
			\node [label=below:{$q$},circle,fill=black,draw=black,inner sep=1pt, minimum size=0.2cm] (q2) at (6-2,7-1.73205080757) {};
			\node [circle, fill=none, draw=none] (e) at (8.6,7) {};

			\draw[-latex, thick] (q0) -- (q1) node [midway, above, fill=none] {$L_{q_0,q}$};
			\draw[-latex] (q1) -- (s1) node [midway, left, fill=none] {$\sigma_1$};
			\draw[-latex] (s1) -- (s2) node [midway, left, fill=none] {$\sigma_2$};
			\draw[-latex] (s2) -- (s3) node [midway, left, fill=none] {$\sigma_3$};
			% \draw[dashed] (s3) -- (q2) node [midway, left, fill=none] {};

			\draw[-latex,dashed,black] ([shift=(225:2cm)]5,7) arc (-135:177:2cm) node [midway, right, fill=none] {$L_{q,q}$};
		\end{tikzpicture}
		}
\end{minipage}

\section{Correctness}%
\label{section-correctness}

The main hurdle in proving the correctness of algorithm $\ATrees$ is to prove Theorem~\ref{theorem-restricted-subset-nbw-dbw}.
The polynomial bound in the proof of Theorem~\ref{theorem-restricted-subset-nbw-dbw} is obtained through a sequence of lemmas
% (Lemmas~\ref{lemma-concatenation} to~\ref{lemma-nbw-dbw-inputs})
bounding the size of the acceptors used in $\ATrees$ subprocedures and the length restrictions and running time in calls to $\RSQ$ made by these procedures. Section~\ref{subsection-bounding-inouts-to-rsq} deals with  bounding the acceptors, and Section~\ref{subsection-length-restrictions-time-bounds} deals with the more challenging part, providing the length restrictions.
%Subsection~\ref{subsection-correctness-of-Atrees}
Finally, Section~\ref{subsection-correctness-of-Atrees} concludes with the theorem stating the correctness of  algorithm $\ATrees$.

\subsection{Bounding the Acceptors}%
\label{subsection-bounding-inouts-to-rsq}

We turn to the representation (as NBW or DBW acceptors) of the languages used in restricted subset queries by $\R^{\omega}(M)$ and its subprocedures. We consider the size, out-degree, and time to construct the acceptors.

In $\R^{\omega}(M)$, there is a restricted subset query with
$M$ itself, and if that query is answered ``no'', a sequence of
restricted subset queries with $M_q$ for accepting states $q$ until an
answer of ``no''.
Clearly, if $M$ is an NBW acceptor, each $M_q$ is
an NBW acceptor of the same size and out-degree and
is easily constructed from $M$, and similarly
if $M$ is a DBW acceptor.

The restricted subset queries made in $\Findctrex$
and its subprocedures are of the form $P \cdot {(S)}^{\omega}$,
where $P$ is a length and prefix restricted version of $L_{q_0,q}$
and $S$ is a concatenation of (at most) a finite word and two length
and prefix restricted versions of $L_{q,q}$.
Therefore in what follows we consider the operations of concatenation and
$\omega$-repetition of regular languages of finite words.

These operations are particularly simple for DFW or NFW
acceptors in \concept{special form}, that is,
containing at most one accepting state, which
has no out-transitions defined.
In general, any NFW acceptor can be converted to special form,
possibly at the cost of increasing its out-degree.
A regular language of finite words is recognized
by a DFW acceptor in special form iff it is
prefix-free.

However, if $M$ is an NBW (resp., DBW) acceptor, then the
finite word languages $L_{q_0,q}$ and $L_{q,q}$ are
recognized by easily constructed
NFW (resp., DFW) acceptors of size at most $|M|$
and out-degree at most the out-degree of $M$.
Lemma~\ref{lemma-length-and-prefix-restriction}
shows that the length and prefix restricted versions
of $L_{q_0,q}$ and $L_{q,q}$ are recognized by
NFW (resp., DFW) acceptors in special form
which may be constructed in time polynomial in
$|M|$, $\ell$, and $|v|$ and have out-degree
at most the out-degree of $M$.

\begin{lem}%
	\label{lemma-concatenation}
	Suppose $M_1$ is an NFW acceptor in special form
	and $M_2$ is an NFW or NBW acceptor.
	Then an acceptor $M$ for $\lang{M_1} \cdot \lang{M_2}$
	can be constructed such that
	\begin{enumerate}
		\item $|M| \le |M_1| + |M_2|$,
		\item the out-degree of $M$ is at most the maximum of out-degrees
		of $M_1$ and $M_2$,
		\item $M$ can be constructed in polynomial time,
		\item $M$ is deterministic if $M_1$ and $M_2$ are deterministic,
		\item $M$ is an NFW in special form if $M_2$ is an NFW in special form.
	\end{enumerate}
\end{lem}

\commentout{

\begin{lem}{\ref{lemma-concatenation}}[restated]
	Suppose $M_1$ is an NFW acceptor in special form
	and $M_2$ is an NFW or NBW acceptor.
	Then an acceptor $M$ for $\lang{M_1} \cdot \lang{M_2}$
	can be constructed such that
	\begin{enumerate}
		\item $|M| \le |M_1| + |M_2|$,
		\item the out-degree of $M$ is at most the maximum of out-degrees
		of $M_1$ and $M_2$,
		\item $M$ can be constructed in polynomial time,
		\item $M$ is deterministic if $M_1$ and $M_2$ are deterministic,
		\item $M$ is an NFW in special form if $M_2$ is an NFW in special form.
	\end{enumerate}
\end{lem}}

\begin{proof}
	Assume the states of $M_1$ and $M_2$ are disjoint.
	If $M_1$ has no accepting state then $\lang{M_1} = \emptyset$
	and we take $M$ to be a one-state acceptor of the same
	kind as $M_2$ that recognizes $\emptyset$.
	Otherwise, $M_1$ has one accepting state $q_1$ with
	no out transitions.
	If $q_1$ is also the initial state of $M_1$, then
	$\lang{M_1} = \{\varepsilon\}$ and we take $M = M_2$.

	Otherwise, $M$ is constructed by taking the union
	of the two machines, removing the state $q_1$ and
	redirecting all the transitions to $q_1$ in $M_1$
	to the initial state of $M_2$.
	The initial state of $M$ is set to be the initial state of $M_1$,
	and the accepting states of $M$ are set to be the accepting
	states of $M_2$.

	Then $M$ is an NFW acceptor if $M_2$ is an NFW
	acceptor, and an NBW acceptor if $M_2$ is an NBW
	acceptor.
	It is straightforward to verify the required properties of $M$.
\end{proof}

\begin{lem}%
	\label{lemma-omega-repetition}
	Suppose $M_1$ is an NFW acceptor in special form.
	Then an NBW acceptor $M$ for $\lang{M_1}^{\omega}$
	can be constructed such that
	\begin{enumerate}
		\item $|M| \le |M_1|$,
		\item the out-degree of $M$ is at most the out-degree of $M_1$,
		\item $M$ can be constructed in polynomial time,
		\item $M$ is deterministic if $M_1$ is deterministic.
	\end{enumerate}
\end{lem}

\commentout{
\begin{lem}{\ref{lemma-omega-repetition}}[restated]
	Suppose $M_1$ is an NFW acceptor in special form.
	Then an NBW acceptor $M$ for $\lang{M_1}^{\omega}$
	can be constructed such that
	\begin{enumerate}
		\item $|M| \le |M_1|$,
		\item the out-degree of $M$ is at most the out-degree of $M_1$,
		\item $M$ can be constructed in polynomial time,
		\item $M$ is deterministic if $M_1$ is deterministic.
	\end{enumerate}
\end{lem}}

\begin{proof}
	If $M_1$ has no accepting states then $\lang{M_1} = \emptyset$.
	Otherwise, $M_1$ has one accepting state with no out transitions.
	If the accepting state of $M_1$ is also its initial state,
	then $\lang{M_1} = \{\varepsilon\}$.
	In these two cases, $\lang{M_1}^{\omega} = \emptyset$ and we take
	$M$ to be an NBW acceptor with one state and no accepting states.

	Otherwise,
	we construct $M$ by removing from $M_1$ its unique accepting
	state $q_1$ and redirecting all the transitions into $q_1$ to the
	initial state of $M_1$.
	The initial state of $M_1$ becomes the unique accepting state
	of $M$.
	It is straightforward to verify the required properties of $M$.
\end{proof}

% The above, together with Lemmas \ref{lemma-concatenation} and \ref{lemma-omega-repetition}, stated in the Appendix, give us the following corollary for the procedure $\R^{\omega}$.
The above give us the following corollary for the procedure $\R^{\omega}$.

\begin{cor}%
	\label{lemma-nbw-dbw-inputs}
	When the input to $\R^{\omega}(M)$ is
	an NBW (resp., DBW) acceptor $M$,
	each $\RSQ$ can be made with
	an NBW (resp., DBW) acceptor
	whose out-degree is at most the out-degree of $M$ and can
	be constructed in time polynomial in $|M|$ and parameters
	giving the length restrictions and the lengths of any words
	that appear.
\end{cor}

\subsection{Length restrictions and time bounds}%
\label{subsection-length-restrictions-time-bounds}

%To prove Theorem~\ref{theorem-restricted-subset-nbw-dbw} we state and prove lemmas
We now turn to establish the correctness and running time
of the subprocedures.
The first two lemmas allow us to bound the parameters giving
the length restrictions in inputs to $\RSQ$.

\begin{lem}%
	\label{lemma-length-bound-k}
	Let $S \subseteq L_{q,q}$ and suppose
	$L_{q_0,q} \cdot {(S)}^{\omega} \setminus L \neq \emptyset$.
	Then
	for some $k < |M|\cdot|T_L|$ we have
	$L_{q_0,q}[k] \cdot {(S)}^{\omega} \setminus L \neq \emptyset$.
\end{lem}

\begin{proof}
	Let $u = \sigma_1 \cdots \sigma_k$ be chosen to be a
	shortest word in $L_{q_0,q}$ such that
	$u \cdot {(S)}^{\omega} \setminus L \neq \emptyset$.
	Then for some $s_1, s_2, \ldots$ from $S$, the $\omega$-word
	\[w = u \cdot s_1 \cdot s_2 \cdots\]
	is in $(L_{q_0,q} \cdot {(S)}^{\omega} \setminus L)$.

	There is an accepting run $r = r_0, r_1, \ldots$ of $M$
	on $w$.
	Let $t = t_0, t_1, \ldots$ be the unique run
	of the DBW acceptor $T_L$ on $w$, which is rejecting.
	Consider the sequence of pairs $(r_n,t_n)$ for $0 \le n \le |u|$.
	If $|u| \ge |M|\cdot|T_L|$, there will be a repeated pair,
	say $(r_i,t_i) = (r_j,t_j)$ for $i < j$.
	If we excise symbols $i+1$ to $j$ of $u$ to get $u'$ and
	the corresponding states from the runs $r$ and $t$ to
	get $r'$ and $t'$, we have
	\[w' = u' \cdot s_1 \cdot s_2 \cdots\]
	is accepted by $M$ (witnessed by $r'$) and rejected by $T_L$
	(witnessed by $t'$), so $u'$ is a shorter word such
	that $u' \cdot {(S)}^{\omega} \setminus L \neq \emptyset$,
	a contradiction.
\end{proof}

% \vspace{-5mm}

\begin{lem}%
	\label{lemma-length-bound-ell}
	Let $S \subseteq L_{q,q}$ and suppose
	$L_{q_0,q} \cdot {(S \cdot L_{q,q})}^{\omega} \setminus L \neq \emptyset$.
	Then for some $k, \ell < |M|\cdot|T_L|$, we have that
	$L_{q_0,q}[k] \cdot {(S \cdot L_{q,q}[\ell])}^{\omega} \setminus L \neq \emptyset$.
\end{lem}

\begin{proof}
	Let $w \in (L_{q_0,q} \cdot {(S \cdot L_{q,q})}^{\omega} \setminus L)$.
	The unique run of the DBW acceptor $T_L$ on $w$
	is rejecting, and does not visit an accepting state of $T_L$
	after some finite prefix.
	Because $S \subseteq L_{q,q}$, we may choose a sufficiently
	long prefix $u$ of $w$ such that $u \in L_{q_0,q}$ and when
	processing $w$, $T_L$ never visits an accepting state after
	reading the prefix $u$.

	Then $w$ may be factored as
	\[w = u (s_1 x_1) (s_2 x_2) \cdots,\]
	where each $s_n \in S$ and each $x_n \in L_{q,q}$.
	There is an accepting run $r = r_0, r_1, \ldots$ of $M$ on $w$,
	which we may assume visits the state $q$ after $u$, and
	also after every $s_n$ and every $x_n$.

	Consider the states $t_1, t_2, \ldots$ visited by $T_L$
	at the start of every group $(s_n x_n)$ when processing $w$.
	After at most $|T_L|$ groups, there must be a repeat,
	say $t_i = t_{i+p}$ for some $p > 0$.
	Let $j = i+p-1$ and consider the $\omega$-word
	\[w' = u \cdot (s_1 x_1) \cdots (s_{i-1} x_{i-1}) \cdot {((s_i x_i) \cdots (s_j x_j))}^{\omega}.\]
	There is an accepting run of $M$ on $w'$, and the unique
	run of $T_L$ on $w'$ is rejecting.
	Let
	\[u' = u \cdot (s_1 x_1) \cdots (s_{i-1} x_{i-1})
	\ \ \textrm{and} \ \
	z = x_i \cdot (s_{i+1} x_{i+1}) \cdots (s_j x_j).\]
	Then $w' = u' \cdot {(s_i z)}^{\omega}$ and $u' \in L_{q_0,q}$ and $z \in L_{q,q}$.

	Consider an accepting run $r' = r_0', r_1', \ldots$ of $M$ on $w'$ that visits
	state $q$ after processing $u'$ and each occurrence of $s_i$ and $z$.
	Consider the unique run $t = t_0', t_1', \ldots$ of $T_L$ on $w'$, which
	is rejecting.
	As in the proof of Lemma~\ref{lemma-length-bound-k},
	if $|z| \ge |M|\cdot|T_L|$ then we may remove a segment
	of $z$ that produces a cycle in the pairs $(r_n',t_n')$.
	Thus, for some $\ell < |M|\cdot|T_L|$, we have
	\[L_{q_0,q} \cdot {(S \cdot L_{q,q}[\ell])}^{\omega} \setminus L \neq \emptyset.\]
	Applying Lemma~\ref{lemma-length-bound-k},
	there also exists $k < |M|\cdot|T_L|$ such
	that
	\[L_{q_0,q}[k] \cdot {(S \cdot L_{q,q}[\ell])}^{\omega} \setminus L \neq \emptyset. \qedhere\]
\end{proof}

We now prove the correctness and polynomial running time of
$\Findprefix$ and $\Findperiod$, which establishes the correctness
and polynomial running time of $\Findctrex$.

\begin{lem}%
	\label{lemma-find-prefix}
	Assume $v \in L_{q,q}$ is such that
	\[L_{q_0,q} \cdot {(v)}^{\omega} \setminus L \neq \emptyset.\]
	Then in time polynomial in $|M|$, $|T_L|$ and $|v|$,
	the procedure $\Findprefix(M,q,v)$ returns a word
	$u \in L_{q_0,q}$ such that
	\[{u(v)}^{\omega} \in (L_{q_0,q} \cdot {(v)}^{\omega} \setminus L).\]
\end{lem}

\begin{proof}
	The algorithm asks restricted subset queries about $L$ for
	$\ell = 0,1,2,\ldots$ to find the least $\ell$ such that
	\[L_{q_0,q}[\ell] \cdot {(v)}^{\omega} \setminus L \neq \emptyset.\]
	The value of $\ell$ is bounded by $|M| \cdot |T_L|$,
	by Lemma~\ref{lemma-length-bound-k}.
	It then searches symbol by symbol for a string $u$
	of length $\ell$ satisfying the required condition.
\end{proof}

% \vspace{-5mm}
The procedure $\Findperiod$ depends on the procedures
$\Nextword$ and $\Nextsymbol$.
The next lemma establishes the correctness and running
time of the procedure $\Nextsymbol$.

\begin{lem}%
	\label{lemma-next-symbol}
	Suppose $\ell$ is a positive integer,
	$y \in L_{q,q}$ or $y = \varepsilon$
	and $v' \in \Sigma^*$ is such that $|v'| < \ell$ and we have
	\[L_{q_0,q} \cdot {(y \cdot L_{q,q}[\ell,v'] \cdot L_{q,q})}^{\omega}
	\setminus L \neq \emptyset.\]
	Then in time polynomial in $|M|$, $|T_L|$, $|y|$ and $\ell$,
	$\Nextsymbol(M,q,y,\ell,v')$ finds
	a symbol $\sigma \in \Sigma$ such that
	\[L_{q_0,q} \cdot {(y \cdot L_{q,q}[\ell, v'\sigma] \cdot L_{q,q})}^{\omega}
	\setminus L \neq \emptyset.\]
\end{lem}
% \vspace{-9mm}
\begin{proof}
	Consider an $\omega$-word
	\[w = u (y v' x_1 y_1) (y v' x_2 y_2) (y v' x_3 y_3) \cdots,\]
	in the language
	\[L_{q_0,q} \cdot {(y \cdot L_{q,q}[\ell,v'] \cdot L_{q,q})}^{\omega}
	\setminus L,\]
	where $u \in L_{q_0,q}$, and for all $i$,
	$v' x_i \in L_{q,q}[\ell,v']$ and $y_i \in L_{q,q}$.
	Fix a particular accepting run of $M_q$ on $w$ that visits
	$q$ after $u$ and after every occurrence of $y$, $v'x_i$ and $y_i$
	in the factorization of $w$ above.

	Because in this run $q$ is visited infinitely many times,
	we may assume that the prefix $u$ is chosen so that $w$ visits
	no accepting state of $T_L$ after the prefix $u$ has been processed.
	Now consider the sequence $t_1, t_2, t_3, \ldots$ of states of $T_L$
	visited by $w$ at the start of every group $(y v' x_i y_i)$.
	This sequence must repeat states of $T_L$, say $t_i = t_{i+p}$ for some
	$p > 0$.
	Let $j = i + p - 1$ and consider the word
	\[w' = u (y v' x_1 y_1) \cdots (y v' x_{i-1} y_{i-1}) {((y v' x_i y_i) \cdots
	(y v' x_j y_j))}^{\omega}.\]
	Clearly, $w' \not\in L$ because after the prefix $u$, $w'$ visits only
	rejecting states of $T_L$.

	Consider the cycle
	\[((y v' x_i y_i) \cdots (y v' x_j y_j)).\]
	If it is of length $1$ (that is $i = j$), then we may duplicate the
	one group $(y v' x_i y_i)$ to make a cycle of length $2$ without
	changing $w'$.
	Then we may factor the cycle as
	\[((y v' x_i y_i) z)
	\ \ \ \ \ \ \textrm{where}\ \ \ \ \ \
	z = (y v' x_{i+1} y_{i+1}) \cdots (y v' x_j y_j)\]
	and $z \in L_{q,q}$.
	Choosing $\sigma$ to be the first symbol of $x_i$ and $x_i'$ to be
	the rest of $x_i$, we have
	\[w' = u' {(y v'\sigma x_i' z)}^{\omega},\]
	where $u' = u (y v' x_1 y_1) \cdots (y v' x_{i-1} y_{i-1})$ and
	therefore
	\[w' \in L_{q_0,q} \cdot
	{(y \cdot L_{q,q}[\ell,v' \sigma] \cdot L_{q,q})}^{\omega}.\]

	Thus we are guaranteed that some symbol $\sigma$ with the required
	property exists.
	Lemma~\ref{lemma-length-bound-ell} (with $S = \{y\} \cdot L_{q,q}[\ell,v'\sigma]$)
	shows that there exist
	$k, m < |M|\cdot|T_L|$ such that
	\[L_{q_0,q}[k] \cdot {(y \cdot L_{q,q}[\ell,v' \sigma] \cdot L_{q,q}[m])}^{\omega}
	\setminus L \neq \emptyset.\]
	Thus, the search for $k$ and $m$ in the procedure $\Nextsymbol$
	can enumerate such pairs $(k,m)$ in increasing order of their
	maximum and try all $\sigma \in \Sigma$ for each pair until
	a suitable symbol $\sigma$ is found to return.
	This process runs in time polynomial in $|M|$, $|T_L|$, $|y|$
	and $\ell$.
\end{proof}

\begin{lem}%
	\label{lemma-Nextword}
	Suppose $y \in L_{q,q}$ or $y = \varepsilon$ is such that
	\[L_{q_0,q} \cdot {(y \cdot L_{q,q})}^{\omega} \setminus L \neq \emptyset.\]
	Then in time bounded by a polynomial in $|M|$, $|T_L|$ and $|y|$,
	$\Nextword(M,q,y)$ returns a word $v' \in L_{q,q}$ of length bounded by
	$|M| \cdot |T_L|$ such that
	\[L_{q_0,q} \cdot {(y v' \cdot L_{q,q})}^{\omega} \setminus L \neq \emptyset.\]
\end{lem}

\begin{proof}
	By Lemma~\ref{lemma-length-bound-ell} (with $S = \{y\}$),
	the search for $k$ and $\ell$ will succeed with both less than
	$|M| \cdot |T_L|$.
	Then $\ell$ calls to the procedure $\Nextsymbol$ will
	produce the required word $v'$ of length $\ell$.
\end{proof}

% \vspace{-3mm}
The next lemma shows that $\Findperiod$ calls $\Nextword$
at most $|T_L|$ times.

\begin{lem}%
	\label{lemma-loop-on-vs}
	Suppose $v_1, v_2, \ldots, v_n \in L_{q,q}$ are such that
	\[L_{q_0,q} \cdot {(v_1 v_2 \cdots v_n \cdot L_{q,q})}^{\omega} \setminus L \neq \emptyset.\]
	Also suppose that the number of states of $T_L$ is less than $n$.
	Then there exist integers $i$ and $j$ with $1 \le i \le j \le n$ such that
	\[L_{q_0,q} \cdot {(v_i v_{i+1} \cdots v_j)}^{\omega} \setminus L \neq \emptyset.\]
\end{lem}
\vspace{-5mm}
\begin{proof}
	Consider an $\omega$-word
	\[w = u (v_1 v_2 \cdots v_n \cdot  y_1) (v_1 v_2 \cdots v_n \cdot y_2)
	(v_1 v_2 \cdots v_n \cdot y_3) \cdots,\]
	in the language
	\[L_{q_0,q} \cdot {(v_1 v_2 \cdots v_n \cdot L_{q,q})}^{\omega} \setminus L,\]
	where $u \in L_{q_0,q}$ and each $y_i \in L_{q,q}$.
	Fix a particular accepting run of $M$ on $w$ in which state $q$ is
	visited after each of the individual segments of $w$.

	Considering the sequence of states of $T_L$
	that are visited in processing $w$,
	there must be some finite prefix after which only
	rejecting states of $T_L$ are visited.
	Because the run of $M$ on $w$ visits $q$ infinitely often,
	we may assume
	that the prefix $u$ of $w$ extends past the last visit
	of $T_L$ to an accepting state.
	Now consider the states $t_1, t_2, \ldots, t_n$ visited by $T_L$
	at the start of each of the first occurrences of $v_1, v_2, \ldots, v_n$,
	respectively.
	Because $n$ is greater than the number of states of $T_L$, some
	state of $T_L$ must repeat in this sequence, say $t_i = t_{i+p}$
	for some $p > 0$.
	Let $j = i+p-1$ and consider the $\omega$-word
	\[w' = u v_1 v_2 \cdots v_{i-1} {(v_i v_{i+1} \cdots v_j)}^{\omega}.\]
	Then $w' \in L_{q_0,q} \cdot {(v_i v_{i+1} \cdots v_j)}^{\omega}$ because
	$u' = u v_1 v_2 \cdots v_{i-1}$ is in $L_{q_0,q}$.
	However, because only rejecting
	states of $T_L$ are visited in the repeating portion of the word,
	$w' \not\in L$.
\end{proof}

The final lemma, presented below, establishes the correctness and polynomial running time
of the procedure $\Findperiod$.

\begin{lem}%
	\label{lemma-Findperiod}
	Suppose $L_{q_0,q} \cdot {(L_{q,q})}^{\omega} \setminus L \neq \emptyset$.
	Then, in polynomial time in $|M|$ and $|T_L|$, the procedure
	$\Findperiod(M,q)$ with restricted query access to $L$
	returns a period word $v$ satisfying the condition
	\[L_{q_0,q} \cdot {(v)}^{\omega} \setminus L \neq \emptyset.\]
\end{lem}

\begin{proof}
	The preconditions of $\Findperiod$ are satisfied, and it
	calls $\Nextword(M,q,y)$ repeatedly, with
	$y = \varepsilon$, then $y = v_1$, then $y = v_1 v_2$, and so
	on, where $v_{n+1}$ is the value returned by the call with
	$y = v_1 v_2 \cdots v_n$.
	Each of these calls satisfies the preconditions of $\Nextword$,
	so after at most $|T_L|$ such calls, $\Findperiod$ returns
	a correct period word $v$, by Lemma~\ref{lemma-loop-on-vs}.
\end{proof}

% \vspace{-5mm}
These lemmas can be used in combination to prove
Theorem~\ref{theorem-restricted-subset-nbw-dbw},
giving a polynomial time reduction of unrestricted subset queries to
restricted subset queries for NBW acceptors
(resp., DBW acceptors.)

\subsection{Correctness of \texorpdfstring{$\ATrees$}{A-trees}}%
\label{subsection-correctness-of-Atrees}

The lemmas established in the previous subsection also show the correctness and running time of $\Findctrex(M',q)$
when called by $\ATrees$,
provided that each $\RSQ$ about $L$ is correctly answered
and $q$ satisfies the precondition of $\Findctrex$.

To complete the consideration of representation issues,
we must prove that $\ATrees$ can successfully simulate
$\Findctrex$ as stated in Lemma~\ref{lemma-regular-tree-inputs}.

\begin{lem}\label{lemma-regular-tree-inputs}
	When $\ATrees$ simulates $\Findctrex(M',q)$ in response
	to a negative counterexample $t$,
	every $\RSQ$ can be simulated with a
	$\MQ$ about $\Trees_d(L)$.
\end{lem}

\begin{proof}
	In the learning algorithm $\ATrees$, when a negative
	counterexample $t$ represented by $A_t$ is received,
	the algorithm simulates the procedure
	$\Findctrex(M',q)$ where $M' = \acceptor(A_t)$ is a
	NBW acceptor recognizing $\paths(t)$ and
	$q$ is an accepting state of $M'$.
	Note that by Lemma~\ref{lemma-tree-nbw},
	$|M'| \le |A_t|$ and
	the out-degree of $M'$ is at most $d$, the arity of $t$.

	Then Corollary~\ref{lemma-nbw-dbw-inputs} shows
	that each $\RSQ$ is with a NBW acceptor that
	has out-degree at most the out-degree of $M'$,
	which is at most $d$.
	Also, each such NBW acceptor
	can be constructed in time polynomial in $|M'|$ and
	parameters giving the length restrictions and the
	lengths of any words that appear.

	The final observation is that
	each such $\RSQ$ is made with an NBW acceptor
	that recognizes a safety language of
	the form $P \cdot {(S)}^{\omega}$, where
	$P$ and $S$ are each languages of fixed-length
	finite words.
	Then, by Lemma~\ref{lemma-nbw-tree}
	each such $\RSQ(N)$
	can be simulated by $\ATrees$ using
	$\MQ(\tree_d(N))$ about $\Trees_d(L)$.
\end{proof}

If $q$ does not satisfy the precondition of $\Findctrex$, then
the procedure may run forever.
However, at least one accepting state $q$ satisfies the
precondition, so at least one simulation will halt and return
$(u,v)$, at which point $\ATrees$ terminates all the other simulations.
This concludes the proof of the reduction given by $\ATrees$,
whose general statement is given in Theorem~\ref{theorem-reduction-general} below.

\begin{thm}%
	\label{theorem-reduction-general}
	Suppose $\class{C} \subseteq \dbw$
	and $\A$ is a polynomial time algorithm that learns class $\class{C}$
	using membership and equivalence queries.
	Then for every positive integer $d$
	there is a polynomial time algorithm $\ATrees$ that learns the class
	$\Trees_d(\class{C})$ using membership and equivalence queries.
\end{thm}

This theorem, together with Maler and Pnueli's~\cite{Maler1995}
polynomial time algorithm to
learn the class of weak regular $\omega$-word languages
using membership and equivalence queries
proves our main result --- Theorem~\ref{theorem-learn-trees}.

\section{Discussion}

We have shown that if
$\class{C} \subseteq \dbw$ can be learned in polynomial time
with membership and equivalence queries, then $\Trees_d(\class{C})$ can
be learned in polynomial time with membership and equivalence
queries for all $d \ge 1$.
Consequently, there is a polynomial time algorithm
to learn $\Trees_d(\dwpw)$ with membership and equivalence
queries.
We have also shown that there are polynomial time algorithms
that implement unrestricted subset queries using restricted subset
queries for $\dfw$, $\nfw$, $\dbw$ and $\nbw$.

One open question is whether there is an interesting
subclass of $\dbw$ that is larger than $\dwpw$ but
still learnable in polynomial time using membership
and equivalence queries,
to which Theorem~\ref{theorem-reduction-general} would also apply.

\section*{Acknowledgment}
The authors would like to thank the anonymous reviewers for their valuable feedback and helpful suggestions.
This research was supported by the United States - Israel Binational Science Foundation, Jerusalem, Israel (BSF) under grant number \#8758451 and by the Office of Naval Research (ONR) under grant number \#N00014-17-1-2787. % chktex 8

\bibliographystyle{alpha}
\bibliography{learn-trees}

%%%%%%%%%%%%%%%%%%%%%%%%%%%%
%%%%% A P P E N D I X %%%%%%
%%%%%%%%%%%%%%%%%%%%%%%%%%%%

\appendix

\commentout{
\section{Figures}

\begin{figure}[h!]
	\centering
	\begin{center}
		\begin{tikzpicture}
			\node [label=below:{$q_0$},circle,fill=black,draw=black,inner sep=1pt, minimum size=0.2cm] (q0) at (0,7) {};
			\node [label=right:{$q$},circle,fill=black,draw=black,inner sep=1pt, minimum size=0.2cm] (q1) at (4,7) {};
			\node [label=below:{$q$},circle,fill=black,draw=black,inner sep=1pt, minimum size=0.2cm] (q2) at (6-1,7-1.73205080757) {};
			\node [label=right:{$q$},circle,fill=black,draw=black,inner sep=1pt, minimum size=0.2cm] (q3) at (7.9,6.4) {};
			\node [label=above:{$q$},circle,fill=black,draw=black,inner sep=1pt, minimum size=0.2cm] (q4) at (6+1.73205080757,7+1) {};
			\node [circle, fill=none, draw=none] (e) at (10,7) {};

			\draw[-latex, thick] (q0) -- (q1) node [midway, above, fill=none] {$L_{q_0,q}$};
			\draw[-latex,thick,black] ([shift=(180:2cm)]6,7) arc (-180:-123:2cm) node [midway, left, fill=none] {$v_1\in L_{q,q}$};
			\draw[-latex,thick,black] ([shift=(240:2cm)]6,7) arc (-120:-21:2cm) node [midway, right, fill=none] {$\,\,v_2\in L_{q,q}$};
			\draw[-latex,thick,black] ([shift=(-20:2cm)]6,7) arc (-20:27:2cm) node [midway, right, fill=none] {$v_3\in L_{q,q}$};
			\draw[-latex,dashed,black] ([shift=(30:2cm)]6,7) arc (30:177:2cm) node [midway, above, fill=none] {$L_{q,q}$};
		\end{tikzpicture}
	\end{center}
\caption{Nextword}
\end{figure}

\begin{figure}[h!]
	\centering
	\begin{center}
		\begin{tikzpicture}
			\node [label=below:{$q_0$},circle,fill=black,draw=black,inner sep=1pt, minimum size=0.2cm] (q0) at (0,7) {};
			\node [label=right:{$q$},circle,fill=black,draw=black,inner sep=1pt, minimum size=0.2cm] (q1) at (4,7) {};
			\node [label=right:{},circle,fill=black,draw=black,inner sep=1pt, minimum size=0.2cm] (s1) at (6-1.93185165258,7-0.5176380902) {};
			\node [label=right:{},circle,fill=black,draw=black,inner sep=1pt, minimum size=0.2cm] (s2) at (6-1.73205080757,7-1) {};
			\node [label=right:{},circle,fill=black,draw=black,inner sep=1pt, minimum size=0.2cm] (s3) at (6-1.41421356237,7-1.41421356237) {};
			\node [label=below:{$q$},circle,fill=black,draw=black,inner sep=1pt, minimum size=0.2cm] (q2) at (6-1,7-1.73205080757) {};
			\node [circle, fill=none, draw=none] (e) at (10,7) {};

			\draw[-latex, thick] (q0) -- (q1) node [midway, above, fill=none] {$L_{q_0,q}$};
			\draw[-latex] (q1) -- (s1) node [midway, left, fill=none] {$\sigma_1$};
			\draw[-latex] (s1) -- (s2) node [midway, left, fill=none] {$\sigma_2$};
			\draw[-latex] (s2) -- (s3) node [midway, left, fill=none] {$\sigma_3$};
			% \draw[dashed] (s3) -- (q2) node [midway, left, fill=none] {};

			\draw[-latex,dashed,black] ([shift=(225:2cm)]6,7) arc (-135:177:2cm) node [midway, right, fill=none] {$L_{q,q}$};
		\end{tikzpicture}
	\end{center}
\caption{Nextsymbol}
\end{figure}
}

\section{Example of \texorpdfstring{$\ATrees$}{A-trees}}%
\label{app:A-tree-example}

We illustrate the algorithm $\ATrees$ learning the language
$\Trees_2(L)$ where $L$ is the language recognized by the DBW
pictured in Fig.~\ref{fig:DBW-for-L}.
Note that $L$ has the rejecting SCC $\{2\}$, which is a subset
of the accepting SCC $\{1,2\}$, so $L$ is not accepted by any
DCW, and is therefore not in $\dbwdcw$.
We assume that the learning algorithm $\ATrees$ has $\MQ$ and $\EQ$
access to the language $\Trees_2(L)$.
$\ATrees$ also has access to an oracle $\A$ that makes $\MQ$s and
$\EQ$s about $L$ and ultimately outputs a DBW recognizing $L$.
The treatment of $\MQ$s is straightforward, so we focus
on $\EQ$s.
To help illustrate the behavior of $\ATrees$, we
choose two hypothetical $\EQ$s that $\A$ could make to $L$,
as well as the possible counterexample trees
to the resulting $\EQ$s made by $\ATrees$.

Suppose the first $\EQ$ that $\A$ makes to $L$
is with the DBW $H_1$ pictured in Fig.~\ref{fig:DBW-H1}.
The language recognized by $H_1$ is $L_1 = (a+b){(a+b+c)}^\omega$, which is
incomparable with $L$.
The algorithm $\ATrees$ constructs the deterministic tree acceptor $H_1^{T,2}$, which
recognizes all binary trees all of whose infinite paths are in $L_1$,
and makes an $\EQ$ to $\Trees_2(L)$ with $H_1^{T,2}$.
Suppose that the counterexample returned is the regular $\omega$-tree
$T_1$ pictured in Fig.~\ref{fig:tree-T1}, with the top three levels of
the extensive form of $T_1$ also shown.

\begin{figure}
  \begin{center}
    \noindent\makebox[\textwidth]{
      \scalebox{0.8}{

        \begin{tikzpicture}[->,>=stealth',shorten >=1pt,auto,node distance=2.2cm,semithick,initial text=,initial where=left]

          \node[label]          (L)                 {$M:$};

          \node[initial,state,accepting](L1) [right  of=L]{$1$};
          \node[state]  (L2)   [right of=L1]{$2$};
          \node[state]  (L3)   [below of=L1]{$3$};

          \path (L1) edge [loop above]  node {$a$} (L1);
          \path (L1) edge [bend left] node {$b$} (L2);
          \path (L1) edge node {$c$} (L3);
          \path (L2) edge [loop above] node {$b$} (L2);
          \path (L2) edge [bend left]  node {$a,c$} (L1);
          \path (L3) edge [loop right]  node {$a,b,c$} (L3);

          \node[label]  (R)   [right of=L, node distance=8cm]  {$H_1:$};

          \node[initial,state]  (H1-1) [right of = R] {$1$};
          \node[state,accepting] (H1-2) [right of = H1-1] {$2$};

          \path (H1-1) edge node {$a,b$} (H1-2);
          \path (H1-2) edge [loop above] node {$a,b,c$} (H1-2);

          \end{tikzpicture}
          }}

    \end{center}
  \caption{Left, the DBW $M$ recognizing $L$. Right, $H_1$, an $\EQ$ made by $\A$ to $L$.}%
  \label{fig:DBW-for-L}%
  \label{fig:DBW-H1}
\end{figure}
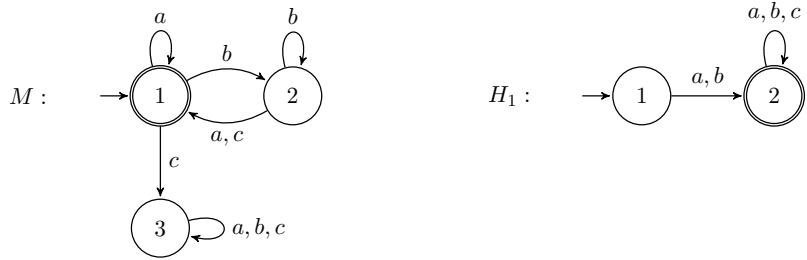

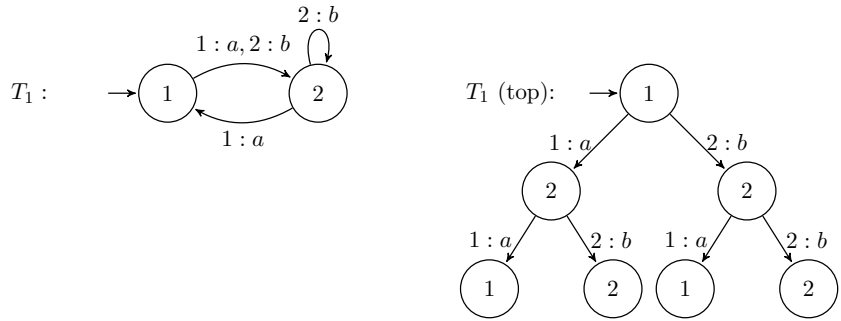
\begin{figure}
  \begin{center}
    \noindent\makebox[\textwidth]{
      \scalebox{0.8}{

        \begin{tikzpicture}[->,>=stealth',shorten >=1pt,auto,node distance=2.3cm,semithick,initial text=,initial where=left]

          \node[label]          (L)                 {$T_1:$};

          \node[initial,state]   (T1-1) [right of = L] {$1$};
          \node[state] (T1-2) [right of = T1-1, node distance=2.5cm] {$2$};

          \path (T1-1) edge [bend left] node {$1:a,2:b$} (T1-2);
          \path (T1-2) edge [loop above] node {$2:b$} (T1-2);
          \path (T1-2) edge [bend left] node {$1:a$} (T1-1);

          \node[label]  (R)   [right of=L, node distance=8cm]  {$T_1$ (top):};

          \node[initial,state]   (X1-1) [right of = R] {$1$};
          \node[state] (X1-2) [below left of = X1-1] {$2$};
          \node[state] (X1-3) [below right of = X1-1] {$2$};
          \node[state] (X1-4) [below left of = X1-2, xshift = .6cm] {$1$};
          \node[state] (X1-5) [below right of = X1-2, xshift = -.6 cm] {$2$};
          \node[state] (X1-6) [below left of = X1-3, xshift = .6cm] {$1$};
          \node[state] (X1-7) [below right of = X1-3, xshift = -.6cm] {$2$};

          \path (X1-1) edge node [left] {$1:a$} (X1-2);
          \path (X1-1) edge node [right] {$2:b$} (X1-3);
          \path (X1-2) edge node [left] {$1:a$} (X1-4);
          \path (X1-2) edge node [right] {$2:b$} (X1-5);
          \path (X1-3) edge node[left]  {$1:a$} (X1-6);
          \path (X1-3) edge node [right] {$2:b$} (X1-7);

    \end{tikzpicture}
    }}

    \end{center}
  \caption{Left, the regular $\omega$-tree counterexample $T_1$ to $H_1^{T,2}$. Right, the top three levels of the extensive form of $T_1$.}%
  \label{fig:tree-T1}
\end{figure}

At this point, the $\ATrees$ algorithm must call on the
$\Acc$ procedure with inputs $T_1$ and $H_1$ to decide whether the
counterexample $T_1$ is accepted or rejected by the hypothesis
$H_1^{T,2}$.
This procedure constructs the product automaton $\Pi_1$ shown in
Fig.~\ref{fig:DBW-product}.
Because $\Pi_1$ accepts ${\{1,2\}}^\omega$, the $\Acc$ procedure
reports that $H_1^{T,2}$ accepts the $\omega$-tree $T_1$.

Because $T_1$ is incorrectly accepted by $H_1^{T,2}$, the learning
algorithm $\ATrees$ constructs a DBW $M' = \acceptor(T_1)$
accepting precisely all the paths of $T_1$.  The DBW $M'$ is shown
in Fig.~\ref{fig:DBW-MP}.
Because at least one $\omega$-word accepted by $M'$ must not
be in the target language $L$, the procedure $\Findctrex$ is
called with the DBW $M'$ and restricted subset query access to the
target language $L$.  The restricted subset queries are simulated
using $\MQ$s to $\Trees_2(L)$ and the representation of a safety language
as a regular $\omega$-tree (Lemma~\ref{lemma-nbw-tree}).

Assume that $\Findctrex$ returns the pair $(a,b)$, representing the
$\omega$-word $ab^\omega$, which is accepted by $H_1$ and is not in $L$.
At this point, the $\EQ$ made by algorithm $\A$ with the DBW $H_1$ can
be answered with the pair $(a,b)$.

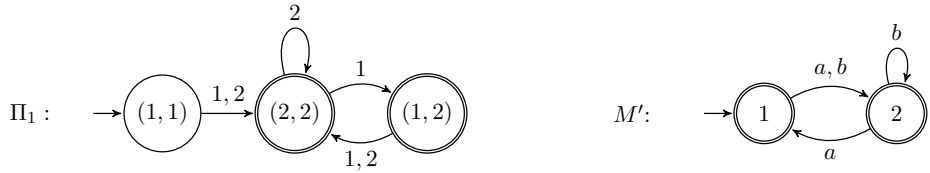
\begin{figure}
  \begin{center}
    \noindent\makebox[\textwidth]{
      \scalebox{0.8}{

        \begin{tikzpicture}[->,>=stealth',shorten >=1pt,auto,node distance=2.2cm,semithick,initial text=,initial where=left]

          \node[label]          (L)                 {$\Pi_1:$};

          \node[initial,state]   (P-1) [right of = L] {$(1,1)$};
          \node[state,accepting] (P-2) [right of = P-1] {$(2,2)$};
          \node[state,accepting] (P-3) [right of = P-2] {$(1,2)$};

          \path (P-1) edge node {$1,2$} (P-2);
          \path (P-2) edge [bend left] node {$1$} (P-3);
          \path (P-2) edge [loop above] node {$2$} (P-2);
          \path (P-3) edge [bend left] node {$1,2$} (P-2);

          \node[label]  (R)   [right of=L, node distance=10cm]  {$M'$:};

          \node[initial,state,accepting] (MP-1) [right of = R] {$1$};
          \node[state,accepting] (MP-2) [right of = MP-1] {$2$};

          \path (MP-1) edge [bend left] node {$a,b$} (MP-2);
          \path (MP-2) edge [bend left] node {$a$} (MP-1);
          \path (MP-2) edge [loop above] node {$b$} (MP-2);

    \end{tikzpicture}
    }}

    \end{center}
  \caption{Left, the product automaton $\Pi_1$ defined using $T_1$ and $H_1$. Right, the DBW $M'= \acceptor(T_1)$ recognizing all the paths in $T_1$.}%
  \label{fig:DBW-product}%
  \label{fig:DBW-MP}
\end{figure}

Assume that at some later point,
$\A$ makes an $\EQ$ to $L$ with the DBW $H_2$ shown in
Fig.~\ref{fig:DBW-H2}.
(Note that the language recognized by $H_2$ is ${(a+ba)}^\omega$, which is a proper subset of $L$.)
$\ATrees$ then makes an $\EQ$ to $\Trees_2(L)$ with the $\omega$-tree automaton
$H_2^{T,2}$.
Assume that the counterexample returned is the regular $\omega$-tree $T_2$,
shown in Fig.~\ref{fig:tree-T2}.
(It can be verified that the tree $T_2$ is in $\Trees_2(L)$.)

Then $\ATrees$ calls $\Acc$ with the tree $T_2$ and the DBW $H_2$.
The $\Acc$ procedure constructs the product DBW $\Pi_2$ shown in
Fig.~\ref{fig:DBW-Pi2}.
The DBW $\Pi_2$ does not accept all $\omega$-words in ${\{1,2\}}^\omega$,
for example, the $\omega$-word $22{(21)}^\omega$ is not accepted.
The corresponding input $\omega$-word is $bc{(ba)}^\omega$, which is
in $L$ but is not accepted by $H_2$.
Thus, the procedure $\Acc$ could return the pair $(bc,ba)$,
which would then be supplied to $\A$ as a counterexample to the $\EQ$
with $H_2$.

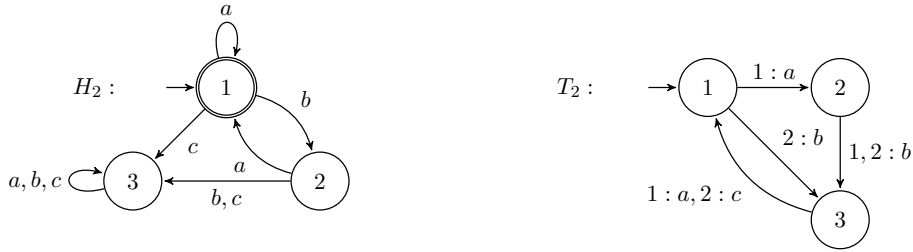
\begin{figure}
  \begin{center}
    \noindent\makebox[\textwidth]{
      \scalebox{0.8}{

        \begin{tikzpicture}[->,>=stealth',shorten >=1pt,auto,node distance=2.2cm,semithick,initial text=,initial where=left]

          \node[label]          (L)                 {$H_2:$};

          \node[initial,state,accepting]  (H2-1) [right of = L] {$1$};
          \node[state] (H2-2) [below right of = H2-1] {$2$};
          \node[state] (H2-3) [below left of = H2-1] {$3$};

          \path (H2-1) edge [loop above] node {$a$} (H2-1);
          \path (H2-1) edge [bend left] node {$b$} (H2-2);
          \path (H2-1) edge  node {$c$} (H2-3);
          \path (H2-2) edge [bend left] node {$a$} (H2-1);
          \path (H2-2) edge node {$b,c$} (H2-3);
          \path (H2-3) edge [loop left] node {$a,b,c$} (H2-3);

          \node[label]  (R)   [right of=L, node distance=8cm]  {$T_2:$};

          \node[initial,state] (T2-1) [right of = R] {$1$};
          \node[state] (T2-2) [right of = T2-1] {$2$};
          \node[state] (T2-3) [below of = T2-2] {$3$};

          \path (T2-1) edge node {$1:a$} (T2-2);
          \path (T2-1) edge node {$2:b$} (T2-3);
          \path (T2-2) edge node {$1,2:b$} (T2-3);
          \path (T2-3) edge [bend left] node {$1:a,2:c$} (T2-1);

          \end{tikzpicture}
          }}

    \end{center}
  \caption{Left, the DBW $H_2$, an $\EQ$ made by $\A$ to $L$.  Right, the counterexample regular $\omega$-tree $T_2$.}%
  \label{fig:DBW-H2}%
  \label{fig:tree-T2}
\end{figure}

\begin{figure}
  \begin{center}
    \noindent\makebox[\textwidth]{
      \scalebox{0.8}{

        \begin{tikzpicture}[->,>=stealth',shorten >=1pt,auto,node distance=2.2cm,semithick,initial text=,initial where=left]

          \node[label]          (L)                 {$\Pi_2:$};

          \node[initial,state,accepting] (Pi2-11) [right of = L] {$(1,1)$};
          \node[state,accepting] (Pi2-21) [above right of = Pi2-11] {$(2,1)$};
          \node[state] (Pi2-32) [below right of = Pi2-11] {$(3,2)$};
          \node[state] (Pi2-13) [right of = Pi2-32] {$(1,3)$};
          \node[state] (Pi2-23) [right of = Pi2-21] {$(2,3)$};
          \node[state] (Pi2-33) [above right of = Pi2-13] {$(3,3)$};

          \path (Pi2-11) edge node  {$1$} (Pi2-21);
          \path (Pi2-11) edge node  {$2$} (Pi2-32);
          \path (Pi2-21) edge node [right] {$1,2$} (Pi2-32);
          \path (Pi2-32) edge [bend left] node {$1$} (Pi2-11);
          \path (Pi2-32) edge node {$2$} (Pi2-13);
          \path (Pi2-13) edge node {$1$} (Pi2-23);
          \path (Pi2-13) edge node {$2$} (Pi2-33);
          \path (Pi2-23) edge node {$1,2$} (Pi2-33);
          \path (Pi2-33) edge [bend left] node {$1,2$} (Pi2-13);

          \end{tikzpicture}
          }}

    \end{center}
  \caption{The product DBW $\Pi_2$ defined using $T_2$ and $H_2$.}%
  \label{fig:DBW-Pi2}
\end{figure}
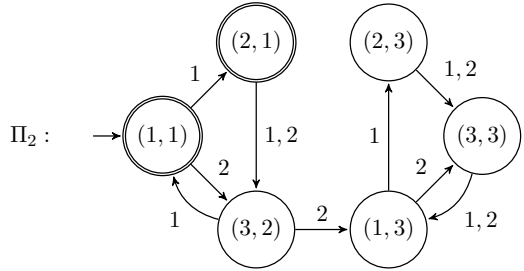

\section{Example of \texorpdfstring{$\Findctrex$}{Findctrex}}%
\label{app:findctrex-example}

Consider the two $\omega$-languages
\[\lang{M_1} = a^*b{((a+c)(b+c))}^\omega,\]
recognized by the DBW $M_1$ pictured in Fig.~\ref{fig:DBW-M1}, and
\[\lang{M_2} = {(a+c)}^*b {((a+c)a^*b)}^\omega,\]
recognized by the DBW $M_2$ pictured in Fig.~\ref{fig:DBW-M2}.
Note that $\lang{M_2}$ is not a subset of $\lang{M_1}$.
For example,
the $\omega$-words
$cb{(ab)}^\omega$
and
$b{(aab)}^\omega$
are both in
$(\lang{M_2} \setminus \lang{M_1})$.

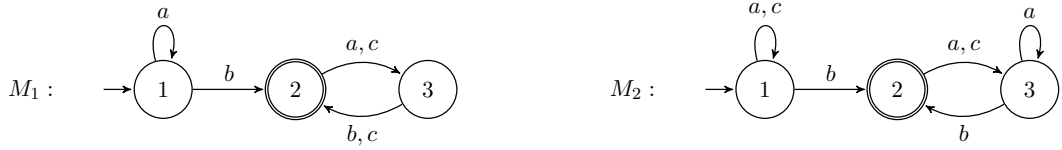
\begin{figure}
  \begin{center}
    \noindent\makebox[\textwidth]{
      \scalebox{0.8}{

        \begin{tikzpicture}[->,>=stealth',shorten >=1pt,auto,node distance=2.2cm,semithick,initial text=,initial where=left]

          \node[label]          (L)                 {$M_1:$};

          \node[initial,state](L1) [right  of=L]{$1$};
          \node[state,accepting]  (L2)   [right of=L1]{$2$};
          \node[state]  (L3)   [right of=L2]{$3$};

          \path (L1) edge [loop above]  node {$a$} (L1);
          \path (L1) edge  node {$b$} (L2);
          \path (L2) edge [bend left]  node {$a,c$} (L3);
          \path (L3) edge [bend left]  node {$b,c$} (L2);

          \node[label]  (R)   [right of=L, node distance=10cm]  {$M_2:$};

          \node[initial,state](R1) [right  of=R]{$1$};
          \node[state,accepting]  (R2)   [right of=R1]{$2$};
          \node[state]  (R3)   [right of=R2]{$3$};

          \path (R1) edge [loop above]  node {$a,c$} (R1);
          \path (R1) edge  node {$b$} (R2);
          \path (R2) edge [bend left]  node {$a,c$} (R3);
          \path (R3) edge [bend left]  node {$b$} (R2);
          \path (R3) edge [loop above] node {$a$} (R3);

          \end{tikzpicture}
          }}

    \end{center}
  \caption{Left, the DBW $M_1$.  Right, the DBW $M_2$.}%
  \label{fig:DBW-M1}%
  \label{fig:DBW-M2}
\end{figure}

Assume that the procedure $\Findctrex$ is called with a restricted
subset query oracle for $\lang{M_1}$ and inputs consisting of
the DBW $M_2$ and final state $2$.
It calls the procedure Findperiod with
inputs $M_2$ and state $2$ to get a
period $v$, for example $v = aab$, with the property that some prefix
followed by ${(aab)}^\omega$ is accepted by $M_2$ and is not in $\lang{M_1}$.
It then calls the procedure Findprefix with inputs
$M_2$, state $2$ and the period $aab$
to find a prefix $u$ for which $u{(aab)}^\omega$ is accepted by $M_2$ and
is not in $\lang{M_1}$, for example $u = b$,
and returns the pair $(b,aab)$ representing the
$\omega$-word $b{(aab)}^\omega$ accepted by $M_2$ but not in $\lang{M_1}$.

In the Findperiod computation on inputs $M_2$ and state $2$, the
Nextword procedure is repeatedly called with inputs $M_2$, state $2$ and
the word $v_1 v_2 \cdots v_{n-1}$ to find the next word $v_n$, until
a restricted subset query yields a period $v = v_i \cdots v_j$ with
the property that some prefix followed by ${(v)}^\omega$ is accepted
by $M_2$ and is not in the language $\lang{M_1}$.
Findperiod returns the period $v$.

For $M_2$ we have the following.
\[\lang{M_2}_{1,2} = {(a+c)}^*b{((a+c)a^*b)}^*
\ \ \ \ \ \ \textrm{and}\ \ \ \ \ \
\lang{M_2}_{2,2} = {((a+c)a^*b)}^+.\]
We consider also the following length-restricted versions of these languages.
\begin{align*}
\lang{M_2}_{1,2}[1] &= b \\
\lang{M_2}_{1,2}[2] &= (a+c)b\\
\lang{M_2}_{1,2}[3] &= (a+c)(a+c)b + b(a+c)b\\
\lang{M_2}_{2,2}[1] &= \emptyset\\
\lang{M_2}_{2,2}[2] &= (a+c)b\\
\lang{M_2}_{2,2}[3] &= (a+c)ab.
\end{align*}

The computation of Nextword on inputs $M_2$, state $2$ and the initial
value $y = \varepsilon$ first searches using restricted subset queries
to find nonnegative integers
$k$ and $\ell$ such that there exist a prefix of length $k$ in $\lang{M_2}_{1,2}$ and
a period of length $\ell$ in $\lang{M_2}_{2,2}$ that yield an $\omega$-word accepted
by $M_2$ that is not in $\lang{M_1}$.  The value of $\ell$ is fixed, and a
word $v'$ of length $\ell$ is built up symbol by symbol calling Nextsymbol
to yield the result of Nextword.
In our example, $k = 1$ and $\ell = 2$ do not suffice
(because $b{((a+c)b)}^\omega$ is a subset of $\lang{M_1}$)
but $k = 2$ and $\ell = 2$ do
(because $(a+c)b{((a+c)b)}^\omega$ is not a subset of $\lang{M_1}$)
and $k = 1$ and $\ell =3$ also do
(because $b{((a+c)ab)}^\omega$ is not a subset of $\lang{M_1}$).
In this example, when $ab$, $cb$, $aab$ or $cab$ is returned by Nextword to Findperiod, the value $v_1$ suffices as the value of $v$ to be returned by Findperiod.
In more complex cases, repeated calls to Nextword may be necessary.

\commentout{

\appendix

\section{Representing a language as paths of a tree: proofs}

\begin{lem}{\ref{lemma-tree-nbw}}[restated]
	If $A_t$ is a regular $\omega$-tree automaton representing an $\omega$-tree $t$,
	then $\paths(t)$ is a safety language recognizable
	by an NBW acceptor $M$ with $|M| = |A_t|$.
\end{lem}

\begin{proof}%[Proof of Lemma~\ref{lemma-tree-nbw}]
	If $A_t = (Q,q_0,\delta,\tau)$, then we define $M = (Q,q_0,\delta',Q)$
	where
	\[\delta'(q,\sigma) = \{r \in Q \mid (\exists i \in D) (\delta(q,i) = r \wedge \tau(q,i) = \sigma)\}\]
	for all $q \in Q$ and $\sigma \in \Sigma$.
	That is, the $M$ transition on $q$ and $\sigma$ is defined to be all states reachable from $q$
	by a transition in $A_t$ labeled with $\sigma$.
	Note that all states of $M$ are accepting.

	If $w \in \paths(t)$, then there is a run $r_0, r_1, \ldots$ of $A_t$
	whose transitions are labeled by $w$, and this is a run of $M$ on $w$, so $w \in \lang{M}$.
	Conversely, if $w \in \lang{M}$, then there is some run $r_0, r_1, \ldots$ of $M$ on $w$,
	and this is a run of $A_t$ whose transitions are labeled with $w$,
	so $w \in \paths(t)$.
\end{proof}

\begin{lem}{\ref{lemma-nbw-tree}}[restated]
	Let $L$ be a safety language recognized by NBW acceptor
	$M = (Q, q_0, \delta, Q)$.
	Suppose the out-degree of $M$ is at most $d$.
	Then there is a $d$-ary regular $\omega$-tree $t$
	such that $\paths(t) = L$, and $t$ is representable by $A_t$ with
	$|A_t| = |M|$.
\end{lem}

\begin{proof}
	We may assume that every state of $M$ is accessible and has at least one
	transition defined.
	We define $A_t = (Q,q_0,\delta_t,\tau)$ over the alphabet $D = \{1,\ldots,d\}$
	as follows.
	For $q \in Q$, choose a surjective mapping $f_q$ from $D$ to $\transitions(q)$.
	Then for $q \in Q$ and $i \in D$, let $(\sigma,r) = f_q(i)$ and define
	$\delta_t(q,i) = r$ and $\tau(q,i) = \sigma$.

	If $w \in L$, then there is a run $r_0, r_1, \ldots$ of $M$ on $w$, and
	there is an infinite path in $A_t$ traversing the same states in which
	the labels are precisely $w$, so $w \in \paths(t)$.
	Conversely, if $w \in \paths(t)$, then there is an infinite path
	$\pi$ such that $t(\pi) = w$, and the sequence of states of $A_t$
	traversed by $w$ yields a run of $M$ on $w$, so $w \in L$.
\end{proof}

\begin{lem}{\ref{lemma-exp-blowup-nbw-dbw}}[restated]
	There exists a family of regular $\omega$-trees $t_1, t_2, \ldots$ such that
	$t_n$ can be represented by a regular $\omega$-tree automaton of size $n+2$,
	but the smallest DBW acceptor recognizing $\paths(t_n)$ has size at
	least $2^n$.
\end{lem}

\begin{proof}
	Let $\Sigma = \{a,b,c\}$ and let $L_n$ be ${(a + b + (a{(a+b)}^n c))}^{\omega}$.
	This is a safety language: $w \in L_n$ iff every occurrence of $c$ in $w$
	is preceded by a word of the form $a{(a+b)}^n$.
	There is a NBW acceptor $M_n$ of $n+2$ states recognizing $L_n$.
	The states are nonnegative integers in $[0,n+1]$, with $0$ the initial
	state,
	$\delta(0,a) = \{0,1\}$, $\delta(0,b) = 0$,
	$\delta(i,a) = \delta(i,b) = i+1$ for $1 \le i \le n$,
	and $\delta(n+1,c) = 0$.

	By Lemma~\ref{lemma-nbw-tree}, there is a ternary regular $\omega$-tree
	$t_n$ such that $\paths(t_n) = L_n$ and $t_n$ is represented by a
	regular $\omega$-tree automaton with $n+2$ states.
	However, any
	DBW acceptor recognizing $L_n$ must have enough states to distinguish
	all $2^n$ strings in ${(a+b)}^n$ in order to check the safety condition.
\end{proof}
}

\commentout{
\section{Restricted subset queries: proofs}\label{section-rsq-proofs}

To prove Theorem~\ref{thm-nfa-dfa-subset-reduction} we first construct an acceptor $M_{\ell,v}$ for
$\lang{M}[\ell,v]$, the length and prefix restricted
version of $\lang{M}$, given $M$, $\ell$ and $v$
as inputs.

\begin{lem}%
	\label{lemma-length-and-prefix-restriction}
	There is a polynomial time algorithm to construct
	an acceptor $M_{\ell,v}$ for $\lang{M}[\ell,v]$ given a NFW acceptor $M$, a nonnegative integer $\ell$ and
	a finite word $v$, such that
	\begin{enumerate}
		\item $M_{\ell,v}$ has at most one accepting state, which has
		no out-transitions,
		\item the out-degree of $M_{\ell,v}$ is at most the out-degree of $M$,
		\item $M_{\ell,v}$ is deterministic if $M$ is deterministic.
	\end{enumerate}
\end{lem}

\commentout{
\begin{lem}{\ref{lemma-length-and-prefix-restriction}}[restated]
	There is a polynomial time algorithm to construct
	an acceptor $M_{\ell,v}$ for $\lang{M}[\ell,v]$ given a nondeterministic
	finite acceptor $M$, a nonnegative integer $\ell$ and
	a finite word $v$, such that
	\begin{enumerate}
		\item $M_{\ell,v}$ has at most one accepting state, which has
		no out-transitions,
		\item the out-degree of $M_{\ell,v}$ is at most the out-degree of $M$,
		\item $M_{\ell,v}$ is deterministic if $M$ is deterministic.
	\end{enumerate}
\end{lem}}

\begin{proof}
	If $\ell < |v|$, then $\lang{M}[\ell,v] = \emptyset$, and the output
	$M_{\ell,v}$ is a one-state acceptor with no accepting states.
	Otherwise, assume $v = \sigma_1 \sigma_2 \cdots \sigma_k$
	and construct $M'$ to be the deterministic finite
	acceptor for $v \cdot \Sigma^{\ell - |v|}$ with states $0, 1, \ldots, \ell$
	where $0$ is the inital state, $\ell$ is the final state, and
	the transitions are $\delta(i,\sigma_{i+1}) = i+1$ for $0 \le i < k$
	and $\delta(i,\sigma) = i+1$ for $k \le i < \ell$ and $\sigma \in \Sigma$.

	Then $M_{\ell,v}$ is obtained by a standard product construction of $M$
	and $M'$ for the intersection $\lang{M} \cap \lang{M'}$, with the
	observation that no accepting state in the product has any
	out-transitions defined, so they may all be identified.
	It is straightforward to verify the required properties of $M_{\ell,v}$.
\end{proof}

\begin{proof}[Proof of Theorem~\ref{thm-nfa-dfa-subset-reduction}]
	For input $M$, define $M_{[\ell,v]}$ to be the finite
	acceptor constructed by the algorithm of
	Lemma~\ref{lemma-length-and-prefix-restriction} to recognize
	the length and prefix restricted language $\lang{M}[\ell,v]$.

	For $\ell = 0, 1, 2, \ldots$,
	ask a restricted subset query with $M_{[\ell,\varepsilon]}$,
	until the first query answered ``no''.
	At this point, $\ell$ is the shortest length
	of a counterexample in $(\lang{M} \setminus L)$.
	Then a counterexample $u$ of length $\ell$ is constructed
	symbol by symbol.

	Assume we have found a prefix $u'$ of a
	counterexample of length $\ell$ in $(\lang{M} \setminus L)$,
	with $|u'| < \ell$.
	For each symbol $\sigma \in \Sigma$
	we ask a restricted subset query with $M_{[\ell,u' \sigma]}$,
	until the first query answered ``no''.
	At this point, $u'$ is extended to $u' \sigma$.
	If the length of $u' \sigma$ is now $\ell$, then $u = u' \sigma$
	is the desired counterexample; otherwise, we
	continue extending $u'$.

	Note that if the input $M$ is deterministic, then all
	of the restricted subset queries are made with deterministic
	finite acceptors.
	If $L$ is recognized by a deterministic finite acceptor $T_L$, then
	the value of $\ell$ is bounded by $|M| \cdot |T_L|$,
	and the algorithm runs in time bounded by a polynomial in
	$|M|$ and $|T_L|$.
\end{proof}
}

\commentout{

\section{Correctness: proofs}

\subsection{Bounding the Acceptors}%
\label{subsection-bounding-inouts-to-rsq}

We turn to the representation (as NBW or DBW acceptors) of the languages used in restricted subset queries by $\R^{\omega}(M)$ and its subprocedures. We consider the size, out-degree, and time to construct the acceptors.

In $\R^{\omega}(M)$, there is a restricted subset query with
$M$ itself, and if that query is answered ``no'', a sequence of
restricted subset queries with $M_q$ for accepting states $q$ until an
answer of ``no''.
Clearly, if $M$ is an NBW acceptor, each $M_q$ is
an NBW acceptor of the same size and out-degree and
is easily constructed from $M$, and similarly
if $M$ is a DBW acceptor.

The restricted subset queries made in $\Findctrex$
and its subprocedures are of the form $P \cdot {(S)}^{\omega}$,
where $P$ is a length and prefix restricted version of $L_{q_0,q}$
and $S$ is a concatenation of (at most) a finite word and two length
and prefix restricted versions of $L_{q,q}$.
Thus we consider the operations of concatenation and
$\omega$-repetition of regular languages of finite words.

These operations are particularly simple for DFW or NFW
acceptors in \concept{special form}, that is,
containing at most one accepting state, which
has no out-transitions defined.
In general, any NFW acceptor can be converted to special form,
possibly at the cost of increasing its out-degree.
A regular language of finite words is recognized
by a DFW acceptor in special form iff it is
prefix-free.

However, if $M$ is an NBW (resp., DBW) acceptor, then the
finite word languages $L_{q_0,q}$ and $L_{q,q}$ are
recognized by easily constructed
NFW (resp., DFW) acceptors of size at most $|M|$
and out-degree at most the out-degree of $M$.
Lemma~\ref{lemma-length-and-prefix-restriction}
shows that the length and prefix restricted versions
of $L_{q_0,q}$ and $L_{q,q}$ are recognized by
NFW (resp., DFW) acceptors in special form
which may be constructed in time polynomial in
$|M|$, $\ell$, and $|v|$ and have out-degree
at most the out-degree of $M$.

\begin{lem}%
	\label{lemma-concatenation}
	Suppose $M_1$ is an NFW acceptor in special form
	and $M_2$ is an NFW or NBW acceptor.
	Then an acceptor $M$ for $\lang{M_1} \cdot \lang{M_2}$
	can be constructed such that
	\begin{enumerate}
		\item $|M| \le |M_1| + |M_2|$,
		\item the out-degree of $M$ is at most the maximum of out-degrees
		of $M_1$ and $M_2$,
		\item $M$ can be constructed in polynomial time,
		\item $M$ is deterministic if $M_1$ and $M_2$ are deterministic,
		\item $M$ is an NFW in special form if $M_2$ is an NFW in special form.
	\end{enumerate}
\end{lem}

\commentout{

\begin{lem}{\ref{lemma-concatenation}}[restated]
	Suppose $M_1$ is an NFW acceptor in special form
	and $M_2$ is an NFW or NBW acceptor.
	Then an acceptor $M$ for $\lang{M_1} \cdot \lang{M_2}$
	can be constructed such that
	\begin{enumerate}
		\item $|M| \le |M_1| + |M_2|$,
		\item the out-degree of $M$ is at most the maximum of out-degrees
		of $M_1$ and $M_2$,
		\item $M$ can be constructed in polynomial time,
		\item $M$ is deterministic if $M_1$ and $M_2$ are deterministic,
		\item $M$ is an NFW in special form if $M_2$ is an NFW in special form.
	\end{enumerate}
\end{lem}}

\begin{proof}
	Assume the states of $M_1$ and $M_2$ are disjoint.
	If $M_1$ has no accepting state then $\lang{M_1} = \emptyset$
	and we take $M$ to be a one-state acceptor of the same
	kind as $M_2$ that recognizes $\emptyset$.
	Otherwise, $M_1$ has one accepting state $q_1$ with
	no out transitions.
	If $q_1$ is also the initial state of $M_1$, then
	$\lang{M_1} = \{\varepsilon\}$ and we take $M = M_2$.

	Otherwise, $M$ is constructed by taking the union
	of the two machines, removing the state $q_1$ and
	redirecting all the transitions to $q_1$ in $M_1$
	to the initial state of $M_2$.
	The initial state of $M$ is the initial state of $M_1$,
	and the accepting states of $M$ are the accepting
	states of $M_2$.

	Then $M$ is an NFW acceptor if $M_2$ is an NFW
	acceptor, and an NBW acceptor if $M_2$ is an NBW
	acceptor.
	It is straightforward to verify the required properties of $M$.
\end{proof}

\begin{lem}%
	\label{lemma-omega-repetition}
	Suppose $M_1$ is an NFW acceptor in special form.
	Then an NBW acceptor $M$ for $\lang{M_1}^{\omega}$
	can be constructed such that
	\begin{enumerate}
		\item $|M| \le |M_1|$,
		\item the out-degree of $M$ is at most the out-degree of $M_1$,
		\item $M$ can be constructed in polynomial time,
		\item $M$ is deterministic if $M_1$ is deterministic.
	\end{enumerate}
\end{lem}

\commentout{
\begin{lem}{\ref{lemma-omega-repetition}}[restated]
	Suppose $M_1$ is an NFW acceptor in special form.
	Then an NBW acceptor $M$ for $\lang{M_1}^{\omega}$
	can be constructed such that
	\begin{enumerate}
		\item $|M| \le |M_1|$,
		\item the out-degree of $M$ is at most the out-degree of $M_1$,
		\item $M$ can be constructed in polynomial time,
		\item $M$ is deterministic if $M_1$ is deterministic.
	\end{enumerate}
\end{lem}}

\begin{proof}
	If $M_1$ has no accepting states then $\lang{M_1} = \emptyset$.
	Otherwise, $M_1$ has one accepting state with no out transitions.
	If the accepting state of $M_1$ is also its initial state,
	then $\lang{M_1} = \{\varepsilon\}$.
	In these two cases, $\lang{M_1}^{\omega} = \emptyset$ and we take
	$M$ to be an NBW acceptor with one state and no accepting states.

	Otherwise,
	we construct $M$ by removing from $M_1$ its unique accepting
	state $q_1$ and redirecting all the transitions into $q_1$ to the
	initial state of $M_1$.
	The initial state of $M_1$ becomes the unique accepting state
	of $M$.
	It is straightforward to verify the required properties of $M$.
\end{proof}

% The above, together with Lemmas \ref{lemma-concatenation} and \ref{lemma-omega-repetition}, stated in the Appendix, give us the following corollary for the procedure $\R^{\omega}$.
The above give us the following corollary for the procedure $\R^{\omega}$.

\begin{cor}%
	\label{lemma-nbw-dbw-inputs}
	When the input to $\R^{\omega}(M)$ is
	an NBW (resp., DBW) acceptor $M$,
	each $\RSQ$ can be made with
	an NBW (resp., DBW) acceptor
	whose out-degree is at most the out-degree of $M$ and can
	be constructed in time polynomial in $|M|$ and parameters
	giving the length restrictions and the lengths of any words
	that appear.
\end{cor}

\subsection{Length restrictions and time bounds}%
\label{subsection-length-restrictions-time-bounds}

%To prove Theorem~\ref{theorem-restricted-subset-nbw-dbw} we state and prove lemmas
We now turn to establish the correctness and running time
of the subprocedures.
The first two lemmas allow us to bound the parameters giving
the length restrictions in inputs to $\RSQ$.

\begin{lem}%
	\label{lemma-length-bound-k}
	Let $S \subseteq L_{q,q}$ and suppose
	$L_{q_0,q} \cdot {(S)}^{\omega} \setminus L \neq \emptyset$.
	Then
	for some $k < |M|\cdot|T_L|$ we have
	$L_{q_0,q}[k] \cdot {(S)}^{\omega} \setminus L \neq \emptyset$.
\end{lem}

\begin{proof}
	Let $u = \sigma_1 \cdots \sigma_k$ be chosen to be a
	shortest word in $L_{q_0,q}$ such that
	$u \cdot {(S)}^{\omega} \setminus L \neq \emptyset$.
	Then for some $s_1, s_2, \ldots$ from $S$, the $\omega$-word
	\[w = u \cdot s_1 \cdot s_2 \cdots\]
	is in $(L_{q_0,q} \cdot {(S)}^{\omega} \setminus L)$.

	There is an accepting run $r = r_0, r_1, \ldots$ of $M$
	on $w$.
	Let $t = t_0, t_1, \ldots$ be the unique run
	of the DBW acceptor $T_L$ on $w$, which is rejecting.
	Consider the sequence of pairs $(r_n,t_n)$ for $0 \le n \le |u|$.
	If $|u| \ge |M|\cdot|T_L|$, there will be a repeated pair,
	say $(r_i,t_i) = (r_j,t_j)$ for $i < j$.
	If we excise symbols $i+1$ to $j$ of $u$ to get $u'$ and
	the corresponding states from the runs $r$ and $t$ to
	get $r'$ and $t'$, we have
	\[w' = u' \cdot s_1 \cdot s_2 \cdots\]
	is accepted by $M$ (witnessed by $r'$) and rejected by $T_L$
	(witnessed by $t'$), so $u'$ is a shorter word such
	that $u' \cdot {(S)}^{\omega} \setminus L \neq \emptyset$,
	a contradiction.
\end{proof}

% \vspace{-5mm}

\begin{lem}%
	\label{lemma-length-bound-ell}
	Let $S \subseteq L_{q,q}$ and suppose
	$L_{q_0,q} \cdot {(S \cdot L_{q,q})}^{\omega} \setminus L \neq \emptyset$.
	Then for some $k, \ell < |M|\cdot|T_L|$, we have that
	$L_{q_0,q}[k] \cdot {(S \cdot L_{q,q}[\ell])}^{\omega} \setminus L \neq \emptyset$.
\end{lem}

\begin{proof}
	Let $w \in (L_{q_0,q} \cdot {(S \cdot L_{q,q})}^{\omega} \setminus L)$.
	The unique run of the DBW acceptor $T_L$ on $w$
	is rejecting, and does not visit an accepting state of $T_L$
	after some finite prefix.
	Because $S \subseteq L_{q,q}$, we may choose a sufficiently
	long prefix $u$ of $w$ such that $u \in L_{q_0,q}$ and when
	processing $w$, $T_L$ never visits an accepting state after
	reading the prefix $u$.

	Then $w$ may be factored as
	\[w = u (s_1 x_1) (s_2 x_2) \cdots,\]
	where each $s_n \in S$ and each $x_n \in L_{q,q}$.
	There is an accepting run $r = r_0, r_1, \ldots$ of $M$ on $w$,
	which we may assume visits the state $q$ after $u$, and
	also after every $s_n$ and every $x_n$.

	Consider the states $t_1, t_2, \ldots$ visited by $T_L$
	at the start of every group $(s_n x_n)$ when processing $w$.
	After at most $|T_L|$ groups, there must be a repeat,
	say $t_i = t_{i+p}$ for some $p > 0$.
	Let $j = i+p-1$ and consider the $\omega$-word
	\[w' = u \cdot (s_1 x_1) \cdots (s_{i-1} x_{i-1}) \cdot {((s_i x_i) \cdots (s_j x_j))}^{\omega}.\]
	There is an accepting run of $M$ on $w'$, and the unique
	run of $T_L$ on $w'$ is rejecting.
	Let
	\[u' = u \cdot (s_1 x_1) \cdots (s_{i-1} x_{i-1})\]
	and
	\[z = x_i \cdot (s_{i+1} x_{i+1}) \cdots (s_j x_j).\]
	Then $w' = u' \cdot {(s_i z)}^{\omega}$ and $u' \in L_{q_0,q}$ and $z \in L_{q,q}$.

	Consider an accepting run $r' = r_0', r_1', \ldots$ of $M$ on $w'$ that visits
	state $q$ after processing $u'$ and each occurrence of $s_i$ and $z$.
	Consider the unique run $t = t_0', t_1', \ldots$ of $T_L$ on $w'$, which
	is rejecting.
	As in the proof of Lemma~\ref{lemma-length-bound-k},
	if $|z| \ge |M|\cdot|T_L|$ then we may remove a segment
	of $z$ that produces a cycle in the pairs $(r_n',t_n')$.
	Thus, for some $\ell < |M|\cdot|T_L|$, we have
	\[L_{q_0,q} \cdot {(S \cdot L_{q,q}[\ell])}^{\omega} \setminus L \neq \emptyset.\]
	Applying Lemma~\ref{lemma-length-bound-k},
	there also exists $k < |M|\cdot|T_L|$ such
	that
	\[L_{q_0,q}[k] \cdot {(S \cdot L_{q,q}[\ell])}^{\omega} \setminus L \neq \emptyset.\]
	%\vspace{-12mm}
\end{proof}

We now prove the correctness and polynomial running time of
$\Findprefix$ and $\Findperiod$, which establishes the correctness
and polynomial running time of $\Findctrex$.

\begin{lem}%
	\label{lemma-find-prefix}
	Assume $v \in L_{q,q}$ is such that
	\[L_{q_0,q} \cdot {(v)}^{\omega} \setminus L \neq \emptyset.\]
	Then in time polynomial in $|M|$, $|T_L|$ and $|v|$,
	the procedure $\Findprefix(M,q,v)$ returns a word
	$u \in L_{q_0,q}$ such that
	\[{u(v)}^{\omega} \in (L_{q_0,q} \cdot {(v)}^{\omega} \setminus L).\]
\end{lem}

\begin{proof}
	The algorithm asks restricted subset queries about $L$ for
	$\ell = 0,1,2,\ldots$ to find the least $\ell$ such that
	\[L_{q_0,q}[\ell] \cdot {(v)}^{\omega} \setminus L \neq \emptyset.\]
	The value of $\ell$ is bounded by $|M| \cdot |T_L|$,
	by Lemma~\ref{lemma-length-bound-k}.
	It then searches symbol by symbol for a string $u$
	of length $\ell$ satisfying the required condition.
\end{proof}

% \vspace{-5mm}
The procedure $\Findperiod$ depends on the procedures
$\Nextword$ and $\Nextsymbol$.
The next lemma establishes the correctness and running
time of the procedure $\Nextsymbol$.

\begin{lem}%
	\label{lemma-next-symbol}
	Suppose $\ell$ is a positive integer,
	$y \in L_{q,q}$ or $y = \varepsilon$
	and $v' \in \Sigma^*$ is such that $|v'| < \ell$ and we have
	\[L_{q_0,q} \cdot {(y \cdot L_{q,q}[\ell,v'] \cdot L_{q,q})}^{\omega}
	\setminus L \neq \emptyset.\]
	Then in time polynomial in $|M|$, $|T_L|$, $|y|$ and $\ell$,
	$\Nextsymbol(M,q,y,\ell,v')$ finds
	a symbol $\sigma \in \Sigma$ such that
	\[L_{q_0,q} \cdot {(y \cdot L_{q,q}[\ell, v'\sigma] \cdot L_{q,q})}^{\omega}
	\setminus L \neq \emptyset.\]
\end{lem}
% \vspace{-9mm}
\begin{proof}
	Consider an $\omega$-word
	\[w = u (y v' x_1 y_1) (y v' x_2 y_2) (y v' x_3 y_3) \cdots,\]
	in
	\[L_{q_0,q} \cdot {(y \cdot L_{q,q}[\ell,v'] \cdot L_{q,q})}^{\omega}
	\setminus L,\]
	where $u \in L_{q_0,q}$, and for all $i$,
	$v' x_i \in L_{q,q}[\ell,v']$ and $y_i \in L_{q,q}$.
	Fix a particular accepting run of $M_q$ on $w$ that visits
	$q$ after $u$ and after every occurrence of $y$, $v'x_i$ and $y_i$
	in the factorization of $w$ above.

	Because in this run $q$ is visited infinitely many times,
	we may assume that the prefix $u$ is chosen so that $w$ visits
	no accepting state of $T_L$ after the prefix $u$ has been processed.
	Now consider the sequence $t_1, t_2, t_3, \ldots$ of states of $T_L$
	visited by $w$ at the start of every group $(y v' x_i y_i)$.
	This sequence must repeat states of $T_L$, say $t_i = t_{i+p}$ for some
	$p > 0$.
	Let $j = i + p - 1$ and consider the word
	\[w' = u (y v' x_1 y_1) \cdots (y v' x_{i-1} y_{i-1}) {((y v' x_i y_i) \cdots
	(y v' x_j y_j))}^{\omega}.\]
	Clearly, $w' \not\in L$ because after the prefix $u$, $w'$ visits only
	rejecting states of $T_L$.

	Consider the cycle
	\[((y v' x_i y_i) \cdots (y v' x_j y_j)).\]
	If it is of length $1$ (that is $i = j$), then we may duplicate the
	one group $(y v' x_i y_i)$ to make a cycle of length $2$ without
	changing $w'$.
	Then we may factor the cycle as
	\[((y v' x_i y_i) z),\]
	where
	\[z = (y v' x_{i+1} y_{i+1}) \cdots (y v' x_j y_j),\]
	and $z \in L_{q,q}$.
	Choosing $\sigma$ to be the first symbol of $x_i$ and $x_i'$ to be
	the rest of $x_i$, we have
	\[w' = u' {(y v'\sigma x_i' z)}^{\omega},\]
	where $u' = u (y v' x_1 y_1) \cdots (y v' x_{i-1} y_{i-1})$ and
	therefore
	\[w' \in L_{q_0,q} \cdot
	{(y \cdot L_{q,q}[\ell,v' \sigma] \cdot L_{q,q})}^{\omega}.\]

	Thus we are guaranteed that some symbol $\sigma$ with the required
	property exists.
	Lemma~\ref{lemma-length-bound-ell} (with $S = \{y\} \cdot L_{q,q}[\ell,v'\sigma]$)
	shows that there exist
	$k, m < |M|\cdot|T_L|$ such that
	\[L_{q_0,q}[k] \cdot {(y \cdot L_{q,q}[\ell,v' \sigma] \cdot L_{q,q}[m])}^{\omega}
	\setminus L \neq \emptyset.\]
	Thus, the search for $k$ and $m$ in the procedure $\Nextsymbol$
	can enumerate such pairs $(k,m)$ in increasing order of their
	maximum and try all $\sigma \in \Sigma$ for each pair until
	a suitable symbol $\sigma$ is found to return.
	This process runs in time polynomial in $|M|$, $|T_L|$, $|y|$
	and $\ell$.
\end{proof}

\begin{lem}%
	\label{lemma-Nextword}
	Suppose $y \in L_{q,q}$ or $y = \varepsilon$ is such that
	\[L_{q_0,q} \cdot {(y \cdot L_{q,q})}^{\omega} \setminus L \neq \emptyset.\]
	Then in time bounded by a polynomial in $|M|$, $|T_L|$ and $|y|$,
	$\Nextword(M,q,y)$ returns a word $v' \in L_{q,q}$ of length bounded by
	$|M| \cdot |T_L|$ such that
	\[L_{q_0,q} \cdot {(y v' \cdot L_{q,q})}^{\omega} \setminus L \neq \emptyset.\]
\end{lem}

\begin{proof}
	By Lemma~\ref{lemma-length-bound-ell} (with $S = \{y\}$),
	the search for $k$ and $\ell$ will succeed with both less than
	$|M| \cdot |T_L|$.
	Then $\ell$ calls to the procedure $\Nextsymbol$ will
	produce the required word $v'$ of length $\ell$.
\end{proof}

% \vspace{-3mm}
The next lemma shows that $\Findperiod$ calls $\Nextword$
at most $|T_L|$ times.

\begin{lem}%
	\label{lemma-loop-on-vs}
	Suppose $v_1, v_2, \ldots, v_n \in L_{q,q}$ are such that
	\[L_{q_0,q} \cdot {(v_1 v_2 \cdots v_n \cdot L_{q,q})}^{\omega} \setminus L \neq \emptyset.\]
	Also suppose that the number of states of $T_L$ is less than $n$.
	Then there exist integers $i$ and $j$ with $1 \le i \le j \le n$ such that
	\[L_{q_0,q} \cdot {(v_i v_{i+1} \cdots v_j)}^{\omega} \setminus L \neq \emptyset.\]
\end{lem}
\vspace{-5mm}
\begin{proof}
	Consider an $\omega$-word
	\[w = u (v_1 v_2 \cdots v_n \cdot  y_1) (v_1 v_2 \cdots v_n \cdot y_2)
	(v_1 v_2 \cdots v_n \cdot y_3) \cdots,\]
	in
	\[L_{q_0,q} \cdot {(v_1 v_2 \cdots v_n \cdot L_{q,q})}^{\omega} \setminus L,\]
	where $u \in L_{q_0,q}$ and each $y_i \in L_{q,q}$.
	Fix a particular accepting run of $M$ on $w$ in which state $q$ is
	visited after each of the individual segments of $w$.

	Considering the sequence of states of $T_L$
	that are visited in processing $w$,
	there must be some finite prefix after which only
	rejecting states of $T_L$ are visited.
	Because the run of $M$ on $w$ visits $q$ infinitely often,
	we may assume
	that the prefix $u$ of $w$ extends past the last visit
	of $T_L$ to an accepting state.
	Now consider the states $t_1, t_2, \ldots, t_n$ visited by $T_L$
	at the start of each of the first occurrences of $v_1, v_2, \ldots, v_n$,
	respectively.
	Because $n$ is greater than the number of states of $T_L$, some
	state of $T_L$ must repeat in this sequence, say $t_i = t_{i+p}$
	for some $p > 0$.
	Let $j = i+p-1$ and consider the $\omega$-word
	\[w' = u v_1 v_2 \cdots v_{i-1} {(v_i v_{i+1} \cdots v_j)}^{\omega}.\]
	Then $w' \in L_{q_0,q} \cdot {(v_i v_{i+1} \cdots v_j)}^{\omega}$ because
	$u' = u v_1 v_2 \cdots v_{i-1}$ is in $L_{q_0,q}$.
	However, because only rejecting
	states of $T_L$ are visited in the repeating portion of the word,
	$w' \not\in L$.
\end{proof}

The final lemma establishes the correctness and polynomial running time
of the procedure $\Findperiod$.

\begin{lem}%
	\label{lemma-Findperiod}
	Suppose $L_{q_0,q} \cdot {(L_{q,q})}^{\omega} \setminus L \neq \emptyset$.
	Then, in polynomial time in $|M|$ and $|T_L|$, the procedure
	$\Findperiod(M,q)$ with restricted query access to $L$
	returns a period word $v$ satisfying the condition
	\[L_{q_0,q} \cdot {(v)}^{\omega} \setminus L \neq \emptyset.\]
\end{lem}

\begin{proof}
	The preconditions of $\Findperiod$ are satisfied, and it
	repeatedly calls $\Nextword(M,q,y)$, with
	$y = \varepsilon$, then $y = v_1$, then $y = v_1 v_2$, and so
	on, where $v_{n+1}$ is the value returned by the call with
	$y = v_1 v_2 \cdots v_n$.
	Each of these calls satisfies the preconditions of $\Nextword$,
	so after at most $|T_L|$ such calls, $\Findperiod$ returns
	a correct period word $v$, by Lemma~\ref{lemma-loop-on-vs}.
\end{proof}

% \vspace{-5mm}
These lemmas combine to prove
Theorem~\ref{theorem-restricted-subset-nbw-dbw},
giving a polynomial time reduction of unrestricted subset queries to
restricted subset queries for NBW acceptors
(resp., DBW acceptors.)

}

% \subsection{Correctness of $\ATrees$}

\commentout{
\begin{lem}{\ref{lemma-regular-tree-inputs}}[restated]
	When $\ATrees$ simulates $\Findctrex(M',q)$ in response
	to a negative counterexample $t$,
	every $\RSQ$ can be simulated with a
	$\MQ$ about $\Trees_d(L)$.
\end{lem}
\vspace{-3mm}
\begin{proof}
	In the learning algorithm $\ATrees$, when a negative
	counterexample $t$ represented by $A_t$ is received,
	the algorithm simulates the procedure
	$\Findctrex(M',q)$ where $M' = \acceptor(A_t)$ is a
	NBW acceptor recognizing $\paths(t)$ and
	$q$ is an accepting state of $M'$.
	Note that by Lemma~\ref{lemma-tree-nbw},
	$|M'| \le |A_t|$ and
	the out-degree of $M'$ is at most $d$, the arity of $t$.

	Then Corollary~\ref{lemma-nbw-dbw-inputs} shows
	that each $\RSQ$ is with a NBW acceptor that
	has out-degree at most the out-degree of $M'$,
	which is at most $d$.
	Also, each such NBW acceptor
	can be constructed in time polynomial in $|M'|$ and
	parameters giving the length restrictions and the
	lengths of any words that appear.

	The final observation is that
	each such $\RSQ$ is made with an NBW acceptor
	that recognizes a safety language of
	the form $P \cdot {(S)}^{\omega}$, where
	$P$ and $S$ are each languages of fixed-length
	finite words.
	Then, by Lemma~\ref{lemma-nbw-tree}
	each such $\RSQ(N)$
	can be simulated by $\ATrees$ using
	$\MQ(\tree_d(N))$ about $\Trees_d(L)$.
\end{proof}
}

\commentout{
\begin{lem}{\ref{lemma-length-bound-k}}[restated]
	Let $S \subseteq L_{q,q}$ and suppose
	$L_{q_0,q} \cdot {(S)}^{\omega} \setminus L \neq \emptyset$.
	Then
	for some $k < |M|\cdot|T_L|$ we have
	$L_{q_0,q}[k] \cdot {(S)}^{\omega} \setminus L \neq \emptyset$.
\end{lem}

\begin{lem}{\ref{lemma-find-prefix}}[restated]
	Assume $v \in L_{q,q}$ is such that
	\[L_{q_0,q} \cdot {(v)}^{\omega} \setminus L \neq \emptyset.\]
	Then in time polynomial in $|M|$, $|T_L|$ and $|v|$,
	the procedure $\Findprefix(M,q,v)$ returns a word
	$u \in L_{q_0,q}$ such that
	\[{u(v)}^{\omega} \in (L_{q_0,q} \cdot {(v)}^{\omega} \setminus L).\]
\end{lem}

\begin{lem}{\ref{lemma-Nextword}}[restated]
	Suppose $y \in L_{q,q}$ or $y = \varepsilon$ is such that
	\[L_{q_0,q} \cdot {(y \cdot L_{q,q})}^{\omega} \setminus L \neq \emptyset.\]
	Then in time bounded by a polynomial in $|M|$, $|T_L|$ and $|y|$,
	$\Nextword(M,q,y)$ returns a word $v' \in L_{q,q}$ of length bounded by
	$|M| \cdot |T_L|$ such that
	\[L_{q_0,q} \cdot {(y \cdot v' L_{q,q})}^{\omega} \setminus L \neq \emptyset.\]
\end{lem}

\begin{lem}{\ref{lemma-Findperiod}}[restated]
	Suppose $L_{q_0,q} \cdot {(L_{q,q})}^{\omega} \neq \emptyset$.
	Then, in polynomial time in $|M|$ and $|T_L|$, the procedure
	$\Findperiod(M,q)$ with restricted query access to $L$
	returns a period word $v$ satisfying the condition
	\[L_{q_0,q} \cdot {(v)}^{\omega} \setminus L \neq \emptyset.\]
\end{lem}}

\end{document}